\UseRawInputEncoding 
\documentclass[english,notitlepage,superscriptaddress,nofootinbib,twocolumn,prl]{revtex4-2}
\usepackage{etoolbox}
\usepackage[T1]{fontenc}
\usepackage[utf8]{inputenc}
\usepackage{amsthm}
\usepackage{amsmath}
\usepackage{amssymb}
\usepackage{graphicx}
\usepackage{subcaption}
\usepackage{braket}
\usepackage[english]{babel}
\usepackage{hyperref}
\hypersetup{colorlinks=true,urlcolor=blue,citecolor = blue}
\usepackage{multirow}
\usepackage{lipsum}
\usepackage{comment}

\usepackage{dsfont}
\usepackage{nicefrac}
\usepackage{float}
\usepackage{siunitx}
\usepackage{bbold}

\usepackage[newcommands]{ragged2e}
\usepackage{graphicx,caption}
\usepackage[font=small, labelfont=bf, justification=justified, singlelinecheck=false]{caption}
\usepackage{xcolor}
\usepackage{hyperref,theoremref}
\hypersetup{
	colorlinks=true,       
	linkcolor=blue,        
	citecolor=red,        
	filecolor=red,     
	urlcolor=blue
}

\definecolor{darkorange}{rgb}{1, 0.55, 0.0}

\definecolor{olivegreen}{rgb}{0,0.5,0.1}

\theoremstyle{plain}
\newtheorem{theorem}{Theorem}[section]
 
\newtheorem{lemma}[theorem]{Lemma}
\newtheorem{corollary}[theorem]{Corollary} 
\newtheorem{definition}[theorem]{Definition}

\begin{document}

\title{A trace distance-based geometric analysis of the stabilizer polytope for few-qubit systems}

\author{Alberto Palhares}
\affiliation{Physics Department, Federal University of Rio Grande do Norte, Natal, 59072-970, Rio Grande do Norte, Brazil}
\affiliation{International Institute of Physics, Federal University of Rio Grande do Norte, 59078-970, Natal, Brazil}
\author{Santiago Zamora}
\affiliation{Physics Department, Federal University of Rio Grande do Norte, Natal, 59072-970, Rio Grande do Norte, Brazil}
\affiliation{International Institute of Physics, Federal University of Rio Grande do Norte, 59078-970, Natal, Brazil}
\author{Rafael A. Mac\^{e}do}
\affiliation{Physics Department, Federal University of Rio Grande do Norte, Natal, 59072-970, Rio Grande do Norte, Brazil}
\affiliation{International Institute of Physics, Federal University of Rio Grande do Norte, 59078-970, Natal, Brazil}
\author{Tailan S. Sarubi}
\affiliation{Physics Department, Federal University of Rio Grande do Norte, Natal, 59072-970, Rio Grande do Norte, Brazil}
\affiliation{International Institute of Physics, Federal University of Rio Grande do Norte, 59078-970, Natal, Brazil}
\author{Joab M. Varela}
\affiliation{Physics Department, Federal University of Rio Grande do Norte, Natal, 59072-970, Rio Grande do Norte, Brazil}
\affiliation{International Institute of Physics, Federal University of Rio Grande do Norte, 59078-970, Natal, Brazil}
\author{Gabriel W. C. Rocha}
\affiliation{Physics Department, Federal University of Rio Grande do Norte, Natal, 59072-970, Rio Grande do Norte, Brazil}
\author{Darlan A. Moreira}
\affiliation{School of Science and Technology, Federal University of Rio Grande do Norte, Natal, Brazil}
\author{Rafael Chaves}
\affiliation{International Institute of Physics, Federal University of Rio Grande do Norte, 59078-970, Natal, Brazil}
\affiliation{School of Science and Technology, Federal University of Rio Grande do Norte, Natal, Brazil}

\begin{abstract}
Non-stabilizerness is a fundamental resource for quantum computational advantage, differentiating classically simulable circuits from those capable of universal quantum computation. Recently, non-stabilizerness has been shown to be relevant for a few qubit systems. In this work, we investigate the geometry of the stabilizer polytope in few-qubit quantum systems, using the trace distance to the stabilizer set to quantify non-stabilizerness. By randomly sampling quantum states, we analyze the distribution of non-stabilizerness for both pure and mixed states and compare the trace distance with other non-stabilizerness measures, as well as entanglement. Additionally, we give an analytical expression for the introduced quantifier, classify Bell-like inequalities corresponding to the facets of the stabilizer polytope, and establish a general concentration result connecting non-stabilizerness and entanglement via Fannes’ inequality. Our findings provide new insights into the geometric structure of non-stabilizerness and its role in small-scale quantum systems, offering a deeper understanding of the interplay between quantum resources. 
\end{abstract}

\maketitle

\section{Introduction}

Non-stabilizerness is a fundamental concept in understanding the potential advantages of quantum algorithms. As shown by the Gottesman-Knill theorem \cite{gottesman1998heisenberg}, stabilizer states \cite{gottesman1997stabilizer,hein2004multiparty} $-$ those that can be generated by Clifford circuits, which consist solely of Hadamard gates, $\pi/2$-phase gates, and controlled-NOT gates $-$ can be efficiently simulated on classical computers \cite{aaronson2004improved}. That is, while entanglement is a necessary resource for quantum computational advantage, it is not sufficient, since Clifford circuits can generate highly entangled states yet remain efficiently classically simulable. Thus, achieving universal quantum computation, in particular in the fault-tolerant regime \cite{bravyi2005universal,bravyi2016improved, howard2017application}, requires accessing states and operations beyond the stabilizer formalism, those that exhibit entanglement and the distinct resource of non-stabilizerness.

Given its significance for quantum information processing, resource theories \cite{Veitch_2014,howard2017application,PhysRevLett.128.210502}, and various measures \cite{Veitch_2014,wang2020efficiently,Hamaguchi2024Robustness,haug2023scalable,cepollaro2024harvesting,cao2024gravitational} of non-stabilizerness have been proposed in the literature, each capturing distinct facets of this resource and providing complementary insights into its role in enabling quantum computational advantage. A widely used measure is the Robustness of Magic (RoM) \cite{howard2017application}, which quantifies how much noise has to be added to the state such that it can be expressed as a convex mixture of stabilizer states. It also measures the complexity of a state using Monte Carlo techniques \cite{heinrich2019robustness, pashayan2015estimating}. Similarly, the mana \cite{Veitch_2014}, derived from the negativity of quasiprobabilities in the Wigner function, establishes a connection between non-stabilizerness and the contextuality of quantum theory~\cite{Delfosse_2017,Heimendahl2019,Zurel2024hiddenvariablemodel}. More recently, stabilizer entropies \cite{leone2022stabilizer,haug2024efficient} have been introduced as an efficient and computationally accessible tool, particularly suited for exploring non-stabilizerness in many-body quantum systems \cite{liu2022many,haug2023quantifying}. These measures offer a versatile framework for characterizing and leveraging non-stabilizerness in diverse quantum information tasks.

A key geometric concept underpinning these ideas is the stabilizer polytope \cite{heinrich2019robustness,obst2024wigner}, a convex polytope whose vertices correspond to pure stabilizer states. For $n$ qudits, the polytope has $ d^n \prod_{i=1}^{n}(d^{i} + 1)$ vertices \cite{Veitch_2014}, reflecting its exponential growth with system size. This complexity makes the exact characterization of its structure intractable for large systems—precisely where non-stabilizerness becomes a vital computational resource beyond classical simulability. While asymptotic properties have been extensively studied in many-body regimes \cite{leone2022stabilizer, liu2022many, chen2024magic, turkeshi2025magic}, comparatively little attention has been given to the few-qubit regime \cite{howard2012nonlocality, howard2015maximum, meyer2024bell, wagner2024certifyingnonstabilizernessquantumprocessors, macedo2025witnessing, iannotti2025entanglement, cusomano2025}. Yet, in this low-dimensional setting, the exponential structure of the polytope becomes fully analyzable: its facets, symmetries, and distance-based properties can be rigorously explored. These features not only deepen our understanding of the polytope itself but also provide essential benchmarks for developing scalable approximations in larger systems. Moreover, geometric analyses in small systems have revealed nontrivial patterns, such as connections between polytope facets and Bell inequalities \cite{howard2012nonlocality, howard2015maximum, meyer2024bell, macedo2025witnessing, cusomano2025}, as well as intricate trade-offs between non-stabilizerness and entanglement \cite{dowling2025bridging}. Recent results also indicate that non-stabilizerness has operational significance beyond computation, playing a role in quantum communication \cite{zamora2025prepare} and randomness generation \cite{vairogs2025extracting}, with nontrivial manifestations even at the single-qubit level. Taken together, these insights highlight that analyzing the geometry of the stabilizer polytope in small systems offers conceptual, computational, and experimental value that extends well beyond their classical simulability.

With this context in mind, this paper aims to provide a detailed analysis of the stabilizer polytope and its key geometric properties in small-dimensional quantum systems. Inspired by techniques commonly used in the study of Bell nonlocality \cite{brito2018quantifying}, we adopt the trace distance as a measure of non-stabilizerness, quantifying the deviation of a given quantum state from the stabilizer polytope. By sampling both pure states according to the Haar measure and mixed ones by the induced Hilbert-Schmidt measure \cite{zyczkowski2001induced, hall1998random}, we gain valuable insights into the structure of the state space. Specifically, we compute the distribution of trace distances and classify Bell-like inequalities associated with the facets of the stabilizer polytope, providing analytical expressions for the simplest cases. Building on these results for small system sizes, we derive a general analytical concentration result, leveraging Fannes’ inequality \cite{fannes1973continuity} to establish a fundamental connection between non-stabilizerness and entanglement.

The paper is organized as follows. In Sec. \ref{sec:geometry_and_magic}, we establish the theoretical framework, defining the stabilizer polytope, the Clifford group, and the measures of non-stabilizerness used in this work. In Sec. \ref{ref:geometric_analysis}, we present our analytical results on the geometry of the polytope, deriving expressions for the NTD and identifying optimal states and Bell-like witnesses. In Sec. \ref{sec:distribution}, we move to a numerical analysis, exploring the distribution of non-stabilizerness and its robustness against noise. In Sec. \ref{sec:magicent}, we investigate the relation between non-stabilizerness and entanglement, providing a general concentration result. Finally, in Sec. \ref{sec:Discussion}, we discuss our results and potential future directions.

\section{Geometry of the Stabilizer Polytope and Non-Stabilizerness}\label{sec:geometry_and_magic}

\subsection{Stabilizer Polytope and Clifford Group}
\label{subsec:stab}

In this work, our landscape of interest is the state space of $n$ qudits, $\mathcal D(\mathcal H_{d,n})$, defined with Hilbert space $\mathcal H_{d,n} = (\mathbb C^d)^{\otimes n}$, with $d$ prime. On it, we will be particularly interested in defining a special class of states, the stabilizer ones \cite{gross2006hudson, gross2007non, Hostens_2005}. This section is dedicated to their construction and to highlighting the corresponding applications.

A set of distinguished operators that act on a qudit are the Heisenberg-Weyl operators $Z$ and $X$, which satisfy $ZX = e^{2\pi i/d}XZ$ and $X^d = Z^d =1$. These operators are unitary, and they form a complete orthogonal basis for the vector space of all linear operators on the Hilbert space. 
 Then, we define the Pauli group for $n$ qudits as:
\begin{equation}
    \mathcal P_{d,n} \equiv \{f_d(\mathbf a, \mathbf b) Z^\mathbf a X^\mathbf b\;|\; \mathbf a, \mathbf b \in \mathbb Z_d^{\times n}\} \;\text{ },\label{eq:Pauli_group}
\end{equation}
where $f_{d\neq 2}(\mathbf a , \mathbf b) = \omega_d^{2^{-1}\mathbf a \cdot \mathbf b }$, where $2^{-1} \mathbf a \cdot \mathbf b$ is computed $\mod d$ with $\omega_d= \exp(2\pi i/d)$, and $f_2(\mathbf a, \mathbf b) = i ^{\mathbf a \cdot \mathbf b}$, with  $Z^\mathbf a  \equiv \otimes_{j=1}^n Z^{a_j}$, $X^\mathbf b \equiv \otimes_{j=1}^n X^{b_j}$. Its elements $D_{\mathbf a , \mathbf b}\equiv f_d(\mathbf a, \mathbf b) Z^\mathbf a X^\mathbf b$ are sometimes referred to as displacement operators.

We know that any pure state can be obtained from the action of a unitary $U \in U(\mathcal H_{d,n})$ on a reference state, which we choose $|0\rangle^{\otimes n} \in \mathcal H_{d,n}$, following the convention in the quantum computing literature. We will be particularly interested on the states generated by what is called the Clifford group $\mathcal C\ell_{d,n} \subseteq U(\mathcal H_{d,n})$, which are defined as the unitaries $C$ which preserve the Pauli group, that is, given some $\mathbf a , \mathbf b \in \mathbb Z_d^{\times n}$, there are $\mathbf a^\prime, \mathbf b^\prime \in \mathbb Z_d^{\times n}$, $\phi \in [0,2\pi)$ such that:
\begin{equation}
    C D_{\mathbf a, \mathbf b}C^\dagger = \exp{(i\phi)} D_{\mathbf a^\prime,\mathbf b^\prime}\;\text{ }.
\end{equation}

Their orbit in the Hilbert space naturally defines the stabilizer (pure) states:
\begin{equation}
    \mathrm{STAB}_{d,n}  \equiv\{C |0\rangle^{\otimes n}\;|\; C \in \mathcal C\ell_{d,n}\} \subseteq \mathcal H_{d,n}\;\text{ }.
\end{equation}

Their name is motivated by having a dual description in terms of abelian groups: A qudit stabilizer group is defined as an abelian subgroup $\mathcal S \subseteq \mathcal P_{d,n}$  on which $e^{i\phi}1 \notin \mathcal S\;;\;\forall \phi \in [0,2\pi)$, where we will also impose that $\mathcal S$ has $n$ generators. Then, there is  a unique state $|\mathcal S\rangle \in \mathcal H_{d,n}$ where $s |\mathcal S\rangle= |\mathcal S \rangle\;,\; \forall s \in \mathcal S$, that is, we say that $|\mathcal S\rangle$ is the state \emph{stabilized} by $\mathcal S$. One can show \cite{gottesman1997stabilizer} that there is a Clifford operation $C_\mathcal S  \in \mathcal C\ell_{d,n}$ such that $|\mathcal S \rangle =C_\mathcal S |0\rangle^{\otimes n}$ and vice-versa: there is a stabilizer group for every choice of a Clifford unitary. Therefore, we will refer to $\mathrm{STAB}_{d,n}$ as both the set of stabilizer states and groups of $n$ qudits.

We can then define stabilizer states as those in $\mathcal D(\mathcal H_{d,n})$ obtained from an ensemble of pure stabilizer ones. This makes sense since $\mathrm{STAB}_{d,n}$ is a countable set, albeit one with exponentially many elements as $n$ increases \cite{aaronson2004improved}. The corresponding convex hull defines what is known as the stabilizer polytope (in its $V$-representation):
\begin{equation}
    \begin{split}
    \label{eq: stab_polytope_V}
        P^\mathrm{STAB}_{d,n} = \Bigg\{\rho \in \mathcal D(\mathcal H_{d,n})\;|\; \rho = \sum_{\mathcal S \in \mathrm{STAB}_{d,n}} p_\mathcal S |\mathcal S \rangle \langle \mathcal S| \; ;\\ \; p_\mathcal S \geq 0 \quad; \quad \sum_{\mathcal S \in \mathrm{STAB}_{d,n}}p_\mathcal S =1\Bigg\}\;\text{ }.
    \end{split}
\end{equation}
We will study how the space of states splits into stabilizer and non-stabilizer states:
\begin{equation}
    \mathcal D(\mathcal H_{d,n})  = P^\mathrm{STAB}_{d,n}\sqcup \mathcal M_{d,n}\;\text{ },
\end{equation}
with respect to the geometry of $P^\mathrm{STAB}$ and its entanglement structure. Given some operator $A \in \mathcal B(\mathcal H_{d,n})$, we will denote $\| A\| \equiv \mathrm{Tr}\left(\sqrt{A^\dagger A}\right)$ as its trace norm.

To provide context for our analysis, we review  known facts of the stabilizer polytope for 1 qubit, 1 qutrit, and 2 qubits. In particular, we show the analytical form of the facets of the polytopes explicitly in the generalized Bloch sphere and we also reproduce the states that maximally violate these inequalities. 

\subsubsection{1 qudit}

For a $d$-dimensional system, the stabilizer polytope $P^\mathrm{STAB}_{d,1}$ can be easily defined by its $H$-representation through the following set of inequalities \cite{Howard2014}:

\begin{equation}
    P^{\text{STAB}}_{d,1} = \{\rho: \text{Tr}(\rho A^{\bf{q}}) \geq 0\text{ },\text{ } {\bf{q}}\in\mathbb{Z}_d^{d+1}\}\text{ },\label{eq: H-rep}
\end{equation}
where $\mathbb{Z}_d ={0,1,...,d-1}$, and $ A^{\bf{q}} = -I_d + \sum_{j=1}^{d+1}\Pi_j^{q_j}$ with $\Pi_j^{q_j}$ denoting the projector onto the eigenvector corresponding to the eigenvalue $\omega^{q_j}$ of the $j$-th operator in the list $\{D_{0,1}, D_{1,0}, D_{1,1},..., D_{1,d-1}\}$. Each operator $D_{x,z}$ is an element of the Pauli group $\mathcal P_{d,1}$ defined in Eq.~(\ref{eq:Pauli_group}). 

\begin{figure}[h!] 
    \centering 
    \includegraphics[width=0.75\columnwidth]{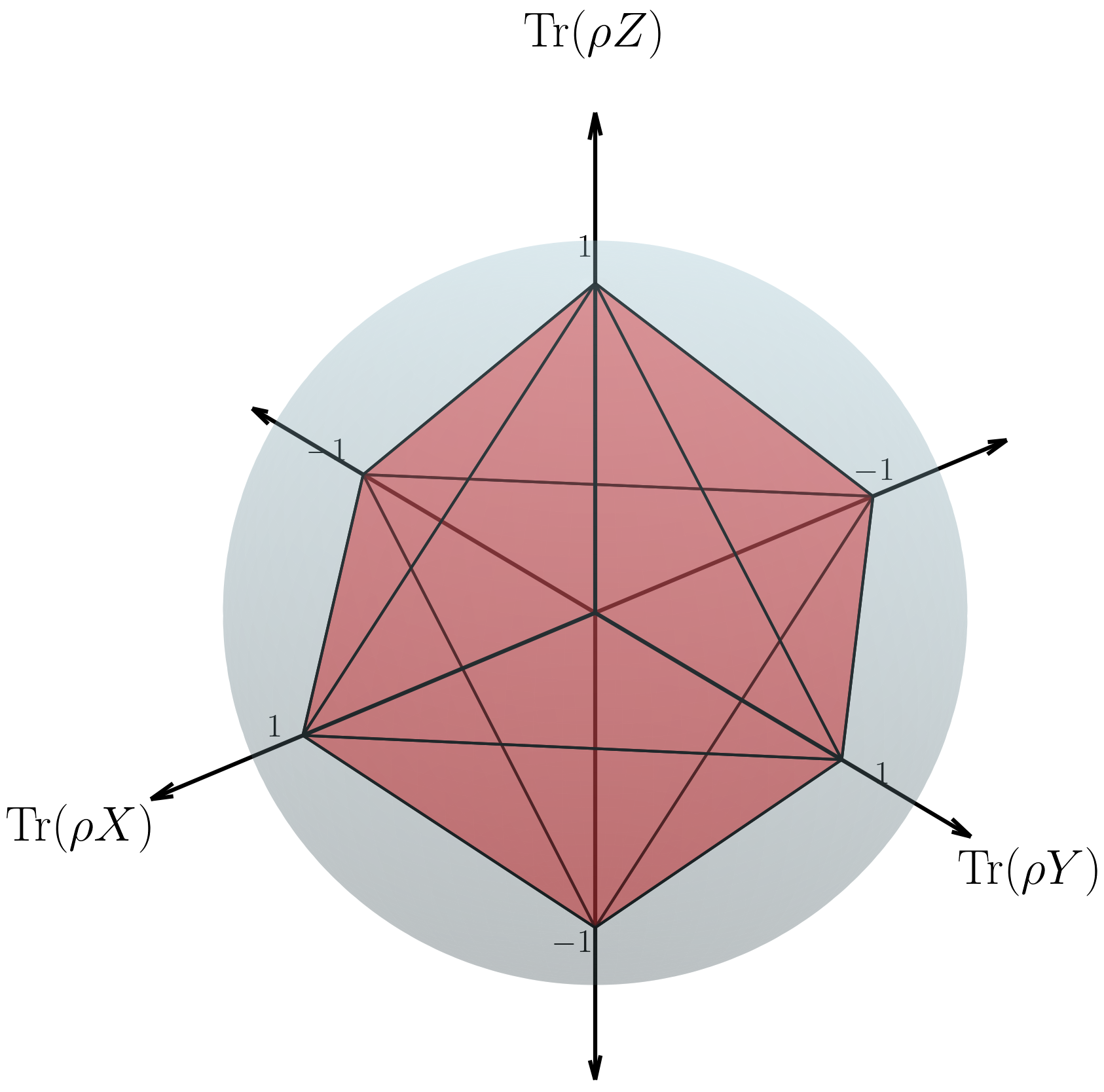} 
    \caption{\raggedright The $P^{\text{STAB}}_{2,1}$ polytope inscribed in the Bloch sphere. }
    \label{fig:BS_stab1} 
\end{figure}

These inequalities can be classified in equivalence classes by applying the natural symmetries of the stabilizer polytope described by Clifford unitaries~\cite{Heinrich2019robustnessofmagic}.  For qudits of  prime dimension $d=p$, Ref.~\cite{PhysRevA.98.032304} shows that the Clifford group is related to the special linear group $SL(2,\mathbb{Z}_p)$, which consists of all  matrices:
\begin{equation}
F= \begin{pmatrix}
a & b \\
c & d
\end{pmatrix}\text{ },
\end{equation}
with elements in  $\mathbb{Z}_p$ and $\det{\left(F\right)} =1$. Remarkably, an element of the Clifford group $C(\boldsymbol{\xi}, \hat{\boldsymbol{F}})$ can be obtained by the expression:
\begin{equation}
    C({\boldsymbol\xi}, \hat{\boldsymbol{F}}) = D_{{\boldsymbol\xi}} V_{\hat{\boldsymbol{F}}}\text{ }, \label{eq:Cliff_qdit}
\end{equation}
where \( D_{\boldsymbol{\xi}} \in \mathcal{P}_{p,1} \) and

\begin{equation}
V_{\hat{F}} =
\begin{cases}
    \frac{1}{\sqrt{p}} \displaystyle \sum_{j,k=0}^{p-1} \omega^{2^{-1}b^{-1}(ak^2 - 2jk + dj^2)} \ket{j} \bra{k}\text{ }, & b \neq 0\text{ }, \\\\
    \frac{1}{\sqrt{p}} \displaystyle\sum_{k=0}^{p-1} \omega^{2^{-1}a c k^2} \ket{ak \bmod d} \bra{k}\text{ }, & b = 0\text{ }.
\end{cases}
\end{equation}
It can be proved that $C\ell_{p,1}$ has $p^3(p^2-1)$ elements. In the following subsections, we will apply these symmetries to the  $P^\mathrm{STAB}_{2,1}$ and $P^\mathrm{STAB}_{3,1}$ polytopes, obtaining the symmetry classes explicitly.
\subsubsection{1 qubit}

\begin{figure}[h!] 
    \centering 
    \includegraphics[width=\columnwidth]{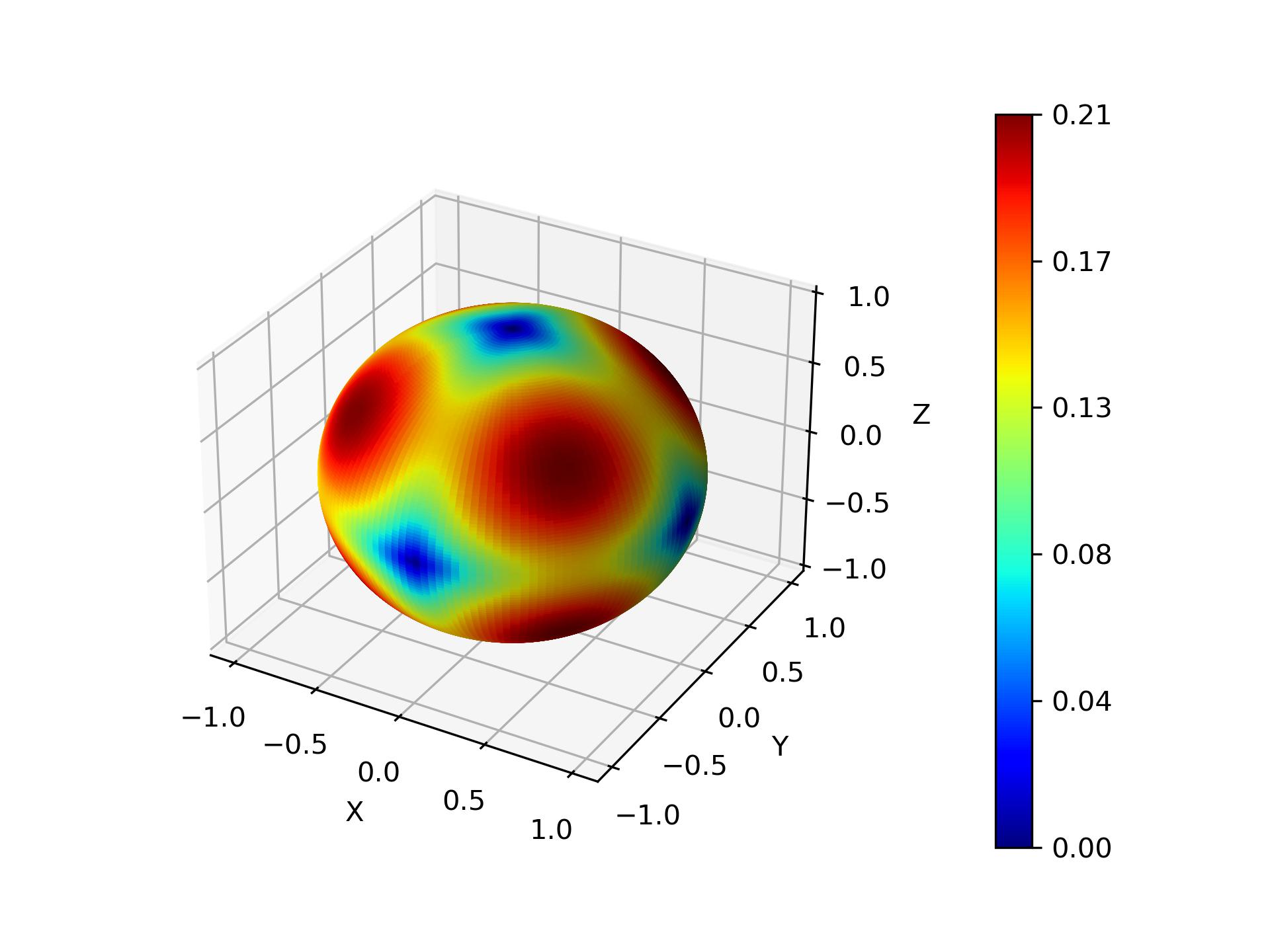} 
    \caption{\raggedright A heat-map of the NTD of the Bloch sphere surface.}
    \label{fig:MTD_heatmap_BS} 
\end{figure}

For qubits, the operators $A^{{\bf q}}$ can be obtained as follows:
\begin{equation}
     \{D_{0,1}, D_{1,0}, D_{1,1}\} = \{Z, X, Y\}\text{ },
\end{equation}
as defined by Eq.~(\ref{eq:Pauli_group}). Note that these operators depend on the projectors of the operators. As an example, we write here one of the eight operators:
\begin{equation}
    A^{(0,0,0)} = -I + \ket{0}\bra{0} + \ket{+}\bra{+} + \ket{i}\bra{i} \text{ }, 
\end{equation}
where  $\ket{0}, \ket{+},\ket{i}$ are the eigenvectors  with eigenvalues $+1$ of the operators $Z,X,Y$ respectively. Each of these eight operators defines a facet of the stabilizer polytope, and by expressing a generic 1-qubit state as:
\begin{equation}   \rho=\sum_{i=0}^3\frac{\braket{\sigma_i}}{2}\sigma_i\quad,\quad\text{with :}\text{ }\begin{cases}
        \sigma_0=I \\
        \sigma_1=X \\
        \sigma_2=Y \\
        \sigma_3=Z
    \end{cases}\quad,
\end{equation}
one can see that the eight facets define an octahedron inscribed in the Bloch sphere (see Fig. \ref{fig:BS_stab1}).

By applying the Clifford unitaries defined in Eq.~(\ref{eq:Cliff_qdit}), each of these facets can be obtained from the inequality:
\begin{equation} \label{eq: 1qbitIneq}
    I_2 = 1 +  \braket{Z} + \braket{X} + \braket{Y} \geq 0 \text{ }. 
\end{equation}
In other words, for the case of 1-qubit, there exists only one equivalence class, illustrated by the representative~(\ref{eq: 1qbitIneq}). 

These inequalities are maximally violated by rotations of the $\ket{T}$ state  (depending on which symmetry is considered) with a density matrix given by:
\begin{equation}
    \rho_T = \frac{1}{2}\left[I+\frac{1}{\sqrt{3}}(X+Y+Z)\right],
\end{equation}
that as mentioned previously has a NTD given by $\mathcal{M}(\rho_T) = 0.211$. In Fig.~\ref{fig:MTD_heatmap_BS}, the red regions represent the states on the surface of the sphere that exhibit the most non-stabilizerness. The state $T$, up to a Clifford unitary, corresponds to the state located at the center of these red areas. In contrast, the blue regions correspond to zones with less non-stabilizerness, with the eigenstates of the Pauli operators situated at their center presenting zero non-stabilizerness.

\subsubsection{1-qutrit}

For $3$ dimensional systems, the operators $A^{{\bf q}}$ in Eq.~(\ref{eq: H-rep}) depend on the projectors of the operators: 
\begin{equation}
     \{D_{0,1}, D_{1,0}, D_{1,1}, D_{1,2}\}\text{ },
\end{equation}
defined by Eq.~(\ref{eq:Pauli_group}). In this situation, there are $81$ inequalities since ${\bf q}\in \mathbb{Z}^4_3$. Writing them analytically does not give much insight into the geometry of the polytope. However, by expressing:
\begin{equation}
\rho=\sum_{i,j=0}^2\frac{\braket{D^\dagger_{i,j}}}{3}D_{i,j}\text{ },
\end{equation}
one can describe the facets of the polytope in terms of the operators $D_{i,j}$ (similarly to the  $1$-qubit case).  With this representation, one can map the inequalities to matrices and work with them numerically. We found that by applying Clifford unitaries [see Eq.~(\ref{eq:Cliff_qdit})], any facet of the polytope can be obtained from the following two inequalities:

\begin{align}
    I^3_1 =&\braket{D^{\dagger}_{00}} +  \braket{D^{\dagger}_{01}} +  \braket{D^{\dagger}_{02}}+ \nonumber\\\ & +\braket{D^{\dagger}_{10}} + \braket{D^{\dagger}_{11}} + \braket{D^{\dagger}_{12}}+    \nonumber\\ 
    & +\braket{D^{\dagger}_{20}} + \braket{D^{\dagger}_{21}}  + \braket{D^{\dagger}_{22}} \geq 0\text{ },\label{eq: Ineq1_qtrits}
\end{align}
\begin{align}
    I^3_2 = &\braket{D^{\dagger}_{00}} +  \braket{D^{\dagger}_{01}} +  \braket{D^{\dagger}_{02}}+ \nonumber\\\ & +\braket{D^{\dagger}_{10}} + \braket{D^{\dagger}_{11}} + \omega\braket{D^{\dagger}_{12}}+    \nonumber\\ 
    & +\braket{D^{\dagger}_{20}} + \omega^* \braket{D^{\dagger}_{21}}  + \braket{D^{\dagger}_{22}} \geq 0\text{ }.\label{eq: Ineq2_qtrits}
\end{align}

Note that these two equivalence classes are also found in the context of  $\Lambda$ polytopes and contextuality \cite{Zurel2024hiddenvariablemodel, Delfosse_2017}. The state that maximally violates Ineq.~(\ref{eq: Ineq1_qtrits}) is the so-called strange state (see Ref.~\cite{Veitch_2014}) given by:
\begin{equation}
    \ket{\mathbb{S}} = \frac{1}{\sqrt{2}}(\ket{1}-\ket{2})\text{ },
\end{equation}
It gives the maximal violation of $I^3_1 = -1$. On the other hand, the state that maximally violates  Ineq.~(\ref{eq: Ineq2_qtrits}) is the pure state:
\begin{equation}
   \ket{\mathbb{S}_2} =  (0.17-i0.07)\ket{0} -(0.67+i0.18)\ket{1}+0.69\ket{2} \text{ },
\end{equation}
which we could not identify in the literature. It reaches the minimum value $I^3_2 =1/2-\sqrt5/2\approx  -0.618$.  

Interestingly, the Norrell state  defined as (see Ref.~\cite{Veitch_2014}):
\begin{equation}
    \ket{\mathbb{N}} = \frac{1}{\sqrt{6}}(-\ket{0}+2\ket{1}-\ket{2})\text{ },
\end{equation}
violates both inequalities by the same amount $I^3_1=I^3_2=-0.5$. Also, it is worth noting that $\ket{\mathbb{S}}$ gives the same violation when evaluating it in the second inequality. In Ref.~\cite{Veitch_2014}, is shown that $\ket{\mathbb{S}}$ and $\ket{\mathbb{N}}$ have the same sum negativity. 

Calculating the NTD for the states $\ket{\mathbb{S}}$, $\ket{\mathbb{S}_2}$ and $\ket{\mathbb{N}}$ we obtain $0.5$, $0.447$, and $1/3$ respectively.
\subsubsection{2-qubits}

For 2-qubits, the facets of the stabilizer polytope can be obtained as follows. Begin by performing a facet enumeration using the LRS software \cite{avis2001lrs, avis2002canonical}, which produces a raw output of $22,320$ inequalities, along with the condition $\braket{I}=1$. For two-qubit systems, the definition in Eq.~(\ref{eq:Cliff_qdit}) no longer applies. Instead, we refer to Ref.~\cite{PhysRevA.87.030301}, where the Clifford unitaries for two-qubit systems are derived and classified as follows:
\begin{enumerate}
    \item \textbf{Single-qubit Class.} This class consists of $576$ elements and is defined by all unitaries of the form:
    \begin{equation}
        (h_0\otimes h_1)(v_0\otimes v_1)(p_0\otimes p_1)\text{ }.
    \end{equation}
    \item \textbf{CNot Class.} This class consists of 5,184 elements, given by the following expression:
    \begin{equation}
        (h_0\otimes h_1)(v_0\otimes v_1)CX_{0,1}(v'_0\otimes v'_1)(p_0\otimes p_1)\text{ }.
    \end{equation}
    \item  \textbf{i-SWAP Class.} This class also contains 5,184 elements, defined by: 
    \begin{equation}
        (h_0\otimes h_1)(v_0\otimes v_1)CX_{0,1}CX_{1,0}(v'_0\otimes v'_1)(p_0\otimes p_1)\text{ }.
    \end{equation}
     \item \textbf{SWAP class.} This class consists of $576$ elements, given by the following expression:
    \begin{equation}
        (h_0\otimes h_1)(v_0\otimes v_1)CX_{0,1}CX_{1,0}CX_{0,1}(p_0\otimes p_1)\text{ }.
    \end{equation}
\end{enumerate}
In the expressions above, $h_0,h_1 \in \{I,H\}, v_0,v_1\in\{I,V,V^2\}$ and $p_0,p_1 \in \{I,X,Y,Z\} $ with $V=(HS)^2$, $H$ being  the Hadamard gate and $S$ the phase gate.  The 2-qubit Clifford group is thus composed of the above $11,520$ elements.

Any element of the 22,320 inequality set can be obtained by applying a Clifford unitary from one of the eight inequalities shown in Table \ref{tab:2Qbits_Stab_ineqs}. This reproduces the $8$ inequalities  obtained in Ref.~\cite{10.5555/2012098.2012105}. They are different representatives of the same classes, related to those we present here by a Clifford transformation. Also, for a description of the polytope from the point of view of $\Lambda$ polytopes, we refer to \cite{Heimendahl2019, ipek2023}.

\begin{table*}[ht!]
\centering
\small
\renewcommand{\arraystretch}{1.5} 
\resizebox{\textwidth}{!}{
\begin{tabular}{|c|c|c|c|c|c|c|c|c|c|c|c|c|c|c|c|c|}
\hline
& $\langle I\otimes I\rangle$ & $\langle I\otimes X\rangle$ & $\langle I\otimes Y\rangle$ & $\langle I\otimes Z\rangle$ &$ \langle X \otimes I\rangle$ & $\langle X \otimes X\rangle$ & $\langle X \otimes Y\rangle$ & $\langle X \otimes Z\rangle$ & $\langle Y \otimes I\rangle$ &$\langle Y \otimes X\rangle$  & $\langle Y \otimes Y\rangle$  & $\langle Y \otimes Z\rangle$ & $\langle Z\otimes I\rangle$ & $\langle Z\otimes X\rangle$ & $\langle Z\otimes Y\rangle$ & $\langle Z\otimes Z\rangle$   \\
\hline
$I^{2,2}_1$ & 1  &  1  &  0  &  0  &  0  &  0  & -1  &  1  &  1  & 1   &  0  &  0  &  0  &  0  & -1  & -1 \\
$I^{2,2}_2$ & 1  &  1  & -1  &  0  &  0  &  0  &  0  &  1  &  0  &  0  &  0  & -1  &  0  &  0  &  0  & -1 \\
$I^{2,2}_3$ & 2  &  1  &  1  &  1  &  1  &  0  &  0  &  2  &  0  &  1  & -1  & 1   & -2  & -1  &  -1  &  -1\\
$I^{2,2}_4$ & 2  &  0  &  2  &  0  &  0  &  0  &  0  &  2  & -1  &  1  & -1  &  1   & -1  & -1  & -1  & -1 \\
$I^{2,2}_5$ & 2  &  0  &  0  & -1  & -1  &  1  & -1  &  2  &  0  &  2  &  0  &  -1  &  0  &  0  & -2  & -1 \\
$I^{2,2}_6$ & 2  &  0  &  1  &  0  &  1  &  1  &  0  & -1  & -1  &  1  &  0  &  1   &  -1 & -1  & 0  & -1 \\
$I^{2,2}_7$ & 3  & -1  &  0  & -2  & -2  &  2  & -1  &  3  &  0  &  2  &  1  &  -1  &  1  &  1  & -2  & -2 \\
$I^{2,2}_8$ & 4  &  3  &  1  &  3  &  2  &  1  & -1  &  3  &  2  &  3  & -1  &   1  & -3  & -2  & -2  & -2\\
\hline
\end{tabular}
}
\caption{\raggedright All facets of the 2-qubit stabilizer polytope can be obtained by applying a Clifford unitary to one of the 8 inequalities presented in this table.}\label{tab:2Qbits_Stab_ineqs}
\end{table*}
In the table, each row $r$ has the coefficients $\alpha_{i,j}$ of the inequality  $I_r^{2,2}$ defined as:
\begin{equation}
    I_r^{2,2} = \sum_{i,j}^4\alpha^r_{i,j} \langle P_i\otimes P_j\rangle  \geq 0\text{ },
\end{equation}
where  $P_0=I, P_1=X, P_2 =Y$ and $P_3 = Z$. As an example, the second row in Table \ref{tab:2Qbits_Stab_ineqs} represents the inequality:
\begin{align}
       \braket{I\otimes X} - \braket{I\otimes Y} + \braket{X\otimes Z}- \nonumber \\ \braket{Y\otimes Z} - \braket{Z\otimes Z} +1\geq 0\text{ }.
\end{align}

\subsection{Quantifying Non-Stabilizerness}\label{subsec:magic}

Following the discussion from the previous section, one needs to explicitly parameterize states in $\mathcal M_{d,n}$ regarding their non-stabilizerness. In this section, we will introduce the NTD measure and compare it with other established measures of non-stabilizerness: RoM \cite{howard2017application}, the stabilizer Rényi entropy (SRE) \cite{leone2022stabilizer}, and also understand its robustness to depolarization noise. 

As will be clear in the following, there are several reasons why we adopt the trace distance to the stabilizer polytope as our measure of non-stabilizerness. First, it offers a clear geometric meaning, quantifying the distinguishability of a quantum state from the set of stabilizer states under optimal measurements. Importantly, unlike many recently proposed measures, the trace distance is naturally defined for both pure and mixed states, allowing us to explore the structure of the full quantum state space. Though computationally intensive for large systems, it remains tractable for the few-qubit scenarios we study here and can be efficiently computed via semidefinite programming. Moreover, we show that this measure connects directly to the violation of facet-defining stabilizer inequalities, providing a natural link between non-stabilizerness quantification and the geometry of the stabilizer polytope.

\begin{figure*}[ht!]
    \centering

    \begin{subfigure}{0.32\textwidth}
        \centering
        \includegraphics[width=\textwidth]{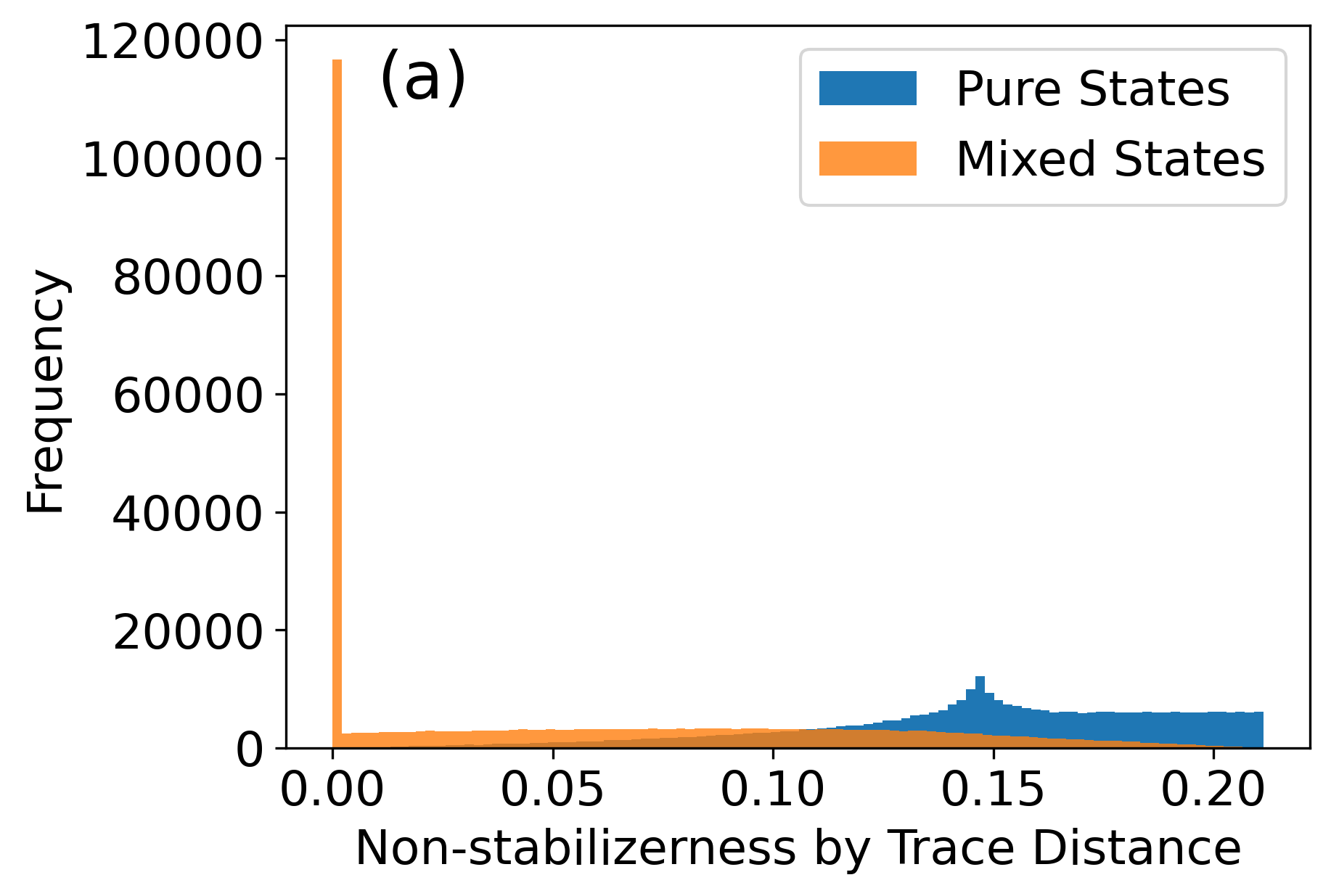}
    \end{subfigure}
    \hfill
    \begin{subfigure}{0.32\textwidth}
        \centering
        \includegraphics[width=\textwidth]{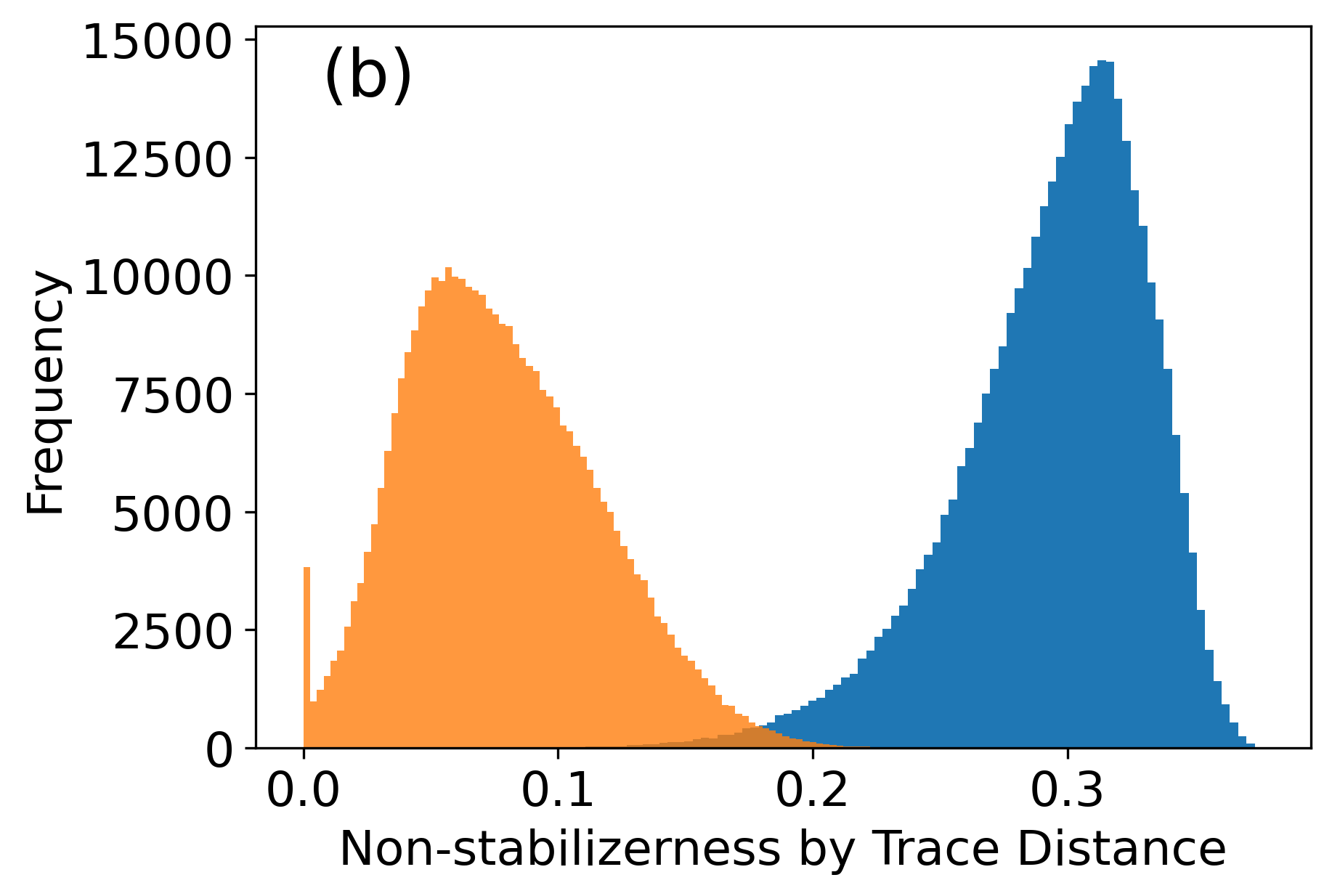}
    \end{subfigure}
    \hfill
    \begin{subfigure}{0.32\textwidth}
        \centering
        \includegraphics[width=\textwidth]{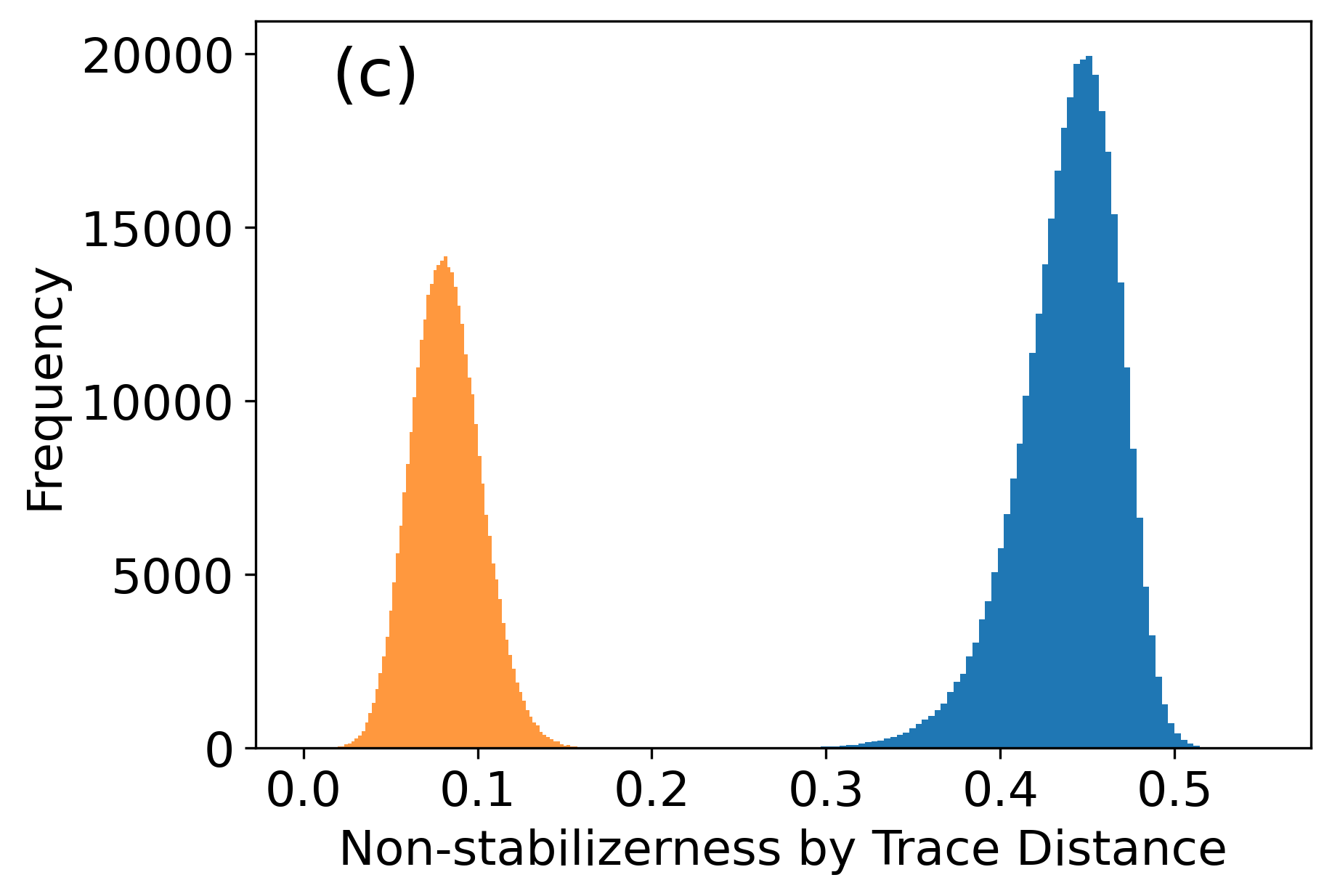}
    \end{subfigure}
   \caption{\raggedright NTD histograms of $360{,}000$ randomly sampled pure (blue) and mixed (orange) states for \textbf{(a)} 1-qubit, \textbf{(b)} 2-qubit, and \textbf{(c)} 3-qubit systems.}

    \label{fig:mtd_hist_all}
\end{figure*}

\subsubsection{Non-stabilizerness by Trace Distance}\label{subsubsec:mtd}

We define the NTD of $\rho$ as the minimum distance from $\rho$ to the stabilizer polytope $P^{\text{STAB}}_{d,n}$:
\begin{equation}
    \mathcal{M}(\rho)=\min_{\sigma\in P^{\text{STAB}}_{d,n}} T(\rho, \sigma) = \min_{\sigma\in P^{\text{STAB}}_{d,n}}\frac {1} {2} \| \rho-\sigma \|_1\text{ }. \label{eq: Magic_TD}
\end{equation}
 By considering Eq.~(\ref{eq: stab_polytope_V}), the numerical problem of determining $\mathcal{M}(\rho)$ for a given $\rho$ can be expressed as the following semi-definite program (SDP) \cite{skrzypczyk2023semidefinite}:
\begin{equation}\label{eq:SDP_MDistance-2}
\begin{split}
 &\min_{\mathbf p = (p_\mathcal S)_\mathcal S} \frac {1} {2} \| \rho-\sigma \|_1\text{ },  \\
\text{subject to}&\\
&\sigma = \sum_{\mathcal S \in \mathrm{STAB}_{d,n}} p_\mathcal S |\mathcal S \rangle \langle \mathcal S|\text{ },\\
& \mathbf p \geq 0\text{ }, \\
& \sum_{\mathcal S \in \mathrm{STAB}_{d,n}}p_\mathcal S = 1\text{ },
\end{split}
\end{equation}
where $\mathbf p \in \mathbb R^{|\mathrm{STAB}_{d,n}|}$ is optimized over $\mathbb R^{|\mathrm{STAB}_{d,n}|}$ and $\mathbf p \geq 0$ denotes element-wise positivity. 
Note that even if the constraints are on the coefficients $p_\mathcal S$, the optimization problem is still an SDP due to the trace norm optimization which its reformulation as an SDP reads \cite{skrzypczyk2023semidefinite}:
\begin{align}
&\|Y\|_1 = \min_{X}  \quad  \, \text{Tr}(X) \\
\text{subject to} & \quad -X \preceq Y \preceq X,\nonumber
\end{align}
where $X$ is an auxiliary positive semi-definite operator, and $Y$ denotes the arbitrary operator whose trace norm is being computed. While this formulation is theoretically direct, we acknowledge that the number of variables $p_S$ grows astronomically with the system size, making this approach computationally intractable in practice.
Nonetheless, as with other trace distance-based resource measures (e.g., in entanglement and Bell nonlocality), the trace distance has a clear and well-motivated geometric interpretation: it quantifies how distinguishable a given state $\rho$ is from the set of stabilizer states under optimal measurements. In other words, it answers the question: What is the stabilizer state $\sigma$ that is hardest to distinguish from $\rho$, when optimizing over all possible measurements?
Furthermore, while the NTD may lack an a priori operational interpretation in terms of simulation cost in the context of quantum computation, it nonetheless offers indirect insight into classical simulability. States with $ \mathcal{M}(\rho) = 0$ lie inside the stabilizer polytope and are thus classically simulable under the Aaronson-Gottesman algorithm \cite{aaronson2004improved}. 

Conversely, as will be discussed later on, a significant amount of noise (through local or global depolarization) is required to push generic non-stabilizer states into the stabilizer polytope, suggesting that NTD can serve as a proxy for the resilience of quantum states against classical simulability.

The NTD inherits all the well-known properties of the trace distance operation~\cite{nielsen2010quantum}, such as its contractiveness under CPTP maps, its invariability under unitary transformations, and its convexity:
\begin{lemma}
   The NTD  satisfies the following properties:
   \begin{enumerate}
    \item Faithfulness: $\mathcal{M}(\rho) = 0$, if and only if $\rho$ is stabilizer, otherwise, $\mathcal{M}(\rho)>0$,
    \item Monotonicity: for all trace-preserving stabilizer channels $\varepsilon$, $\mathcal{M}(\varepsilon(\rho)) \leq \mathcal{M}(\rho)$,
    \item Convexity: $\mathcal{M}(\sum_i p_i\rho_i) \leq \sum_i p_i\mathcal{M}(\rho_i)$, for $p_i \geq 0$ and $\sum_i p_i = 1$.
    \item Invariance under Clifford Unitaries: $\mathcal{M}(C\rho C^{\dagger}) = \mathcal{M}(\rho)$, for any Clifford unitary $C$.
\end{enumerate}
\end{lemma}

These properties ensure that the NTD is a well-behaved resource monotone, consistent with the framework of quantum resource theories. The proofs for each statement are provided in Appendix \ref{apd:mtd_prop}. In the next section, we use the  NTD measure to analyze the distribution of non-stabilizer states, and we compare it with other well-known measures of non-stabilizerness. 

While the NTD offers a direct geometric interpretation, it is instructive to compare it against established quantifiers. Below, we briefly review the RoM and SRE, which will serve as benchmarks in our numerical analysis.

\subsubsection{Robustness of Magic} \label{subsubsec:Robustness_of_Magic}
The development of the resource theory of non-stabilizerness can be approached similarly to the well-established resource theory of entanglement, particularly concerning the robustness of entanglement, as investigated in \cite{vidal1999robustness}. In that context, the robustness of a quantum state quantifies how much separable noise can be introduced to the system to transform the state into a separable one. This concept, however, can be extended to non-stabilizerness \cite{howard2017application} as seen below:
\begin{definition} \label{def:RoM}  
     The robustness of magic of a state $\rho$ is defined by:
    \begin{equation} \label{eq:RoM_definition}
        \mathcal{R}(\rho) := \min_{\mathbf x}{\left\{ \sum_{\mathcal S \in \mathrm{STAB}_{d,n}} |x_\mathcal S|\text{ };\text{ } \rho = \sum_{\mathcal S} x_\mathcal S |\mathcal S\rangle \langle \mathcal S| \right\}}\text{ },
    \end{equation}
    where $\sum_\mathcal S x_\mathcal S = 1$, $\mathbf x  \in \mathbb R^{|\mathrm{STAB}_{d,n}|}$ and $\sum_{\mathcal S} x_\mathcal S |\mathcal S\rangle \langle \mathcal S|$ is called a stabilizer pseudomixture, which means that $x_\mathcal S$ can be negative.
\end{definition}
    When $\mathcal{R}(\rho) = 1$, $\rho$ is a stabilizer state, otherwise, $\rho$ has non-stabilizerness. Computationally, the RoM can be found as the solution to the following LP:
    \begin{equation}
        \mathcal{R}(\rho) = \min_{{\bf x}}{||\mathbf{x}||_1}\quad \textrm{ s.t. } Ax = b\text{ },
    \end{equation}
    where $||{\bf x}||_1 = \sum_\mathcal S |x_\mathcal S|$, $b_i = \text{Tr}(P_i\rho)$, $A_{i,j} = \text{Tr}(P_j\sigma_i)$, and $P_i$ is $i$-th Pauli word for $n$-qubits.

\subsubsection{Stabilizer Rényi entropies} \label{subsubsec:Quantum_Rényi_Entropies}
In his influential 1961 work, Alfréd Rényi, building on prior contributions by Feinstein and Fadeev, presented an axiomatic framework for deriving Shannon entropy \citep{shannon1948mathematical}. Rényi demonstrated that by imposing five intuitive conditions on functionals defined over a probability space, one uniquely identifies Shannon entropy. By relaxing one of these conditions, he introduced a broader class of entropies now known as Rényi entropies. The stabilizer Rényi entropies, as Ref. \cite{leone2022stabilizer} puts, is a measure of nonstabilizerness associated with the probability of a given Pauli String representing a given pure state.
\begin{definition} \label{def:Renyi_entropy}

    The stabilizer Rényi entropy of order \( \alpha \in (0, 1) \cup (1, \infty) \) is defined as:
    \begin{equation} \label{eq:Renyi_entropy-definition}
        M_{\alpha}(\ket{\psi}) := \frac{1}{1 - \alpha} \left\{ \log \left[ \sum_{P \in \tilde{\mathcal{P}}_{2,n}} \left(\Xi_{P}(\ket{\psi})\right)^{\alpha} \right]\right\} - \log(2^n),
    \end{equation}
    with $n$ being the number of qubits, $\Xi_{P}$ being
    \begin{equation}
        \Xi_{P}(\ket{\psi}) := \frac{1}{2^{n}}\bra{\psi} P \ket{\psi}^2\text{ },
    \end{equation}
    and $\tilde{\mathcal{P}}_{2,n}$ being the group of all $n$-qubit Pauli strings with $+1$ phases.
\end{definition}

\subsection{Entanglement Classification}\label{subsec:entanglement_class}

The relationship between entanglement and non-stabilizerness is subtle, as these resources capture different aspects of quantum advantage. Highly entangled states can possess zero non-stabilizerness, making them classically simulable despite their strong quantum correlations. Conversely, separable states, also classically accessible, can lie outside the stabilizer polytope, exhibiting non-stabilizerness despite the absence of entanglement. In fact, as we will show in Sec. \ref{sec:distribution}, for few-qubit systems, the state with the highest non-stabilizerness can be separable. Notwithstanding, in the asymptotic limit, when the number of qubits is taken to be large, random states become typical, and several results can be derived due to concentration of measure results \cite{mele2024introduction}.

A plethora of results exists depending on how the randomness is imbued in the state. For a Haar-random state, it is known that both the generic entanglement and non-stabilizerness measures are saturated \cite{hayden2006aspects,leone2022stabilizer, liu2022many}. However, non-stabilizerness measures appear to be typically saturated on settings where entanglement measures are not, such as random matrix product states \cite{chen2024magic}, or in random quantum circuits with $\log n$ depth, as numerically found in \cite{turkeshi2025magic}. Hence, even in the typical regime, the relation between entanglement and non-stabilizerness can have wildly different scalings.

In order to understand the finer structure of their relationship, we once again rely on random samples of $2$ and $3$-qubit states to gain insight into the relation between entanglement and non-stabilizerness. This was previously studied in \cite{iannotti2025entanglement}, when states are uniformly distributed. But, as we will see, states with low non-stabilizerness have limited ranges for their amount of entanglement. Based on this observation, we will prove a general continuity result, stating that states close to the stabilizer polytope should inherit the entanglement properties of the stabilizer states.

Let us briefly recall the notion of entanglement classes. Two states $|\psi_1\rangle, |\psi_2\rangle \in \mathcal H$ are said to be in the same LU (\textbf{L}ocal \textbf{U}nitary) entanglement class (denoted by $|\psi_1\rangle \sim_{\mathrm{LU}} |\psi_2\rangle$) if it exists single-qudit unitaries $U_1, U_2, \ldots, U_n$ such that that transforms $|\psi_2\rangle$ into $|\psi_1\rangle$. That is:
\begin{equation}
    |\psi_1\rangle \sim_{\mathrm{LU}} |\psi_2\rangle \quad :\Leftrightarrow \quad |\psi_1\rangle = U_1 \otimes U_2\otimes \cdots \otimes U_n|\psi_2\rangle\;.
\end{equation}

To see why this definition is sound, notice that if one starts with a product state $|\psi_1 \rangle \otimes |\psi_2 \rangle \otimes \cdots \otimes |\psi_n \rangle$, the unitary action will create another product state. The idea is that the Hilbert space $\mathcal H$ splits into a set of classes, denoted by $[\mathcal H]_\mathrm{LU}$, on which each class is a set of states that can be transformed into each other by local unitaries.

It is known that there are continuous families of entanglement classes, with a rich classification \cite{walter2016multipartite}. If the set of pure states is restricted to be the stabilizers, only a discrete and finite set of states $\mathrm{STAB}_{d,n}$ needs to be classified. 

Since Clifford unitaries are the natural transformations of stabilizer states, one can also define the equivalence class $\sim_\mathrm{LC}$, where two quantum states are equivalent if they are related by a local Clifford. Although it is known that the unitary and Clifford classifications are distinct in general, meaning that $[\mathrm{STAB}_{d,n}]_\mathrm{LC} \neq [\mathrm{STAB}_{d,n}]_\mathrm{LU}$ \cite{van2005local,ji2007lu, claudet2025local}, it is known that the LC classification of stabilizer states is equivalent to the LU one in the case for $n\leq 3$ qudits \cite{tripartite2011looi}. We will focus on the case of $n \leq 3$ qubits.

\subsubsection{2-qubits}
The first result concerns the classification for the entanglement structure of $2$-qubits, already analyzed in Ref. \cite{fattal2004entanglement}. It is shown that there are two $\mathrm{LU}$ classes:
\begin{equation}
    [\mathrm{STAB}_{2,2}]_\mathrm{LC} = \left\{[|00\rangle]_\mathrm{LC} \text{ }\text{ }, \text{ }\left[ \frac{1}{\sqrt 2}(|00\rangle + |11\rangle)\right]_\mathrm{LC}\right\}\text{ },
\end{equation}
meaning that every single stabilizer state can be written as:
\begin{equation}
    (C_1 \otimes C_2) |00\rangle\quad\text{or}\quad(C_1 \otimes C_2) \frac{1}{\sqrt{2}}(|00\rangle+|11\rangle)\text{ }.
\end{equation}
That is, either is a state separable or has maximal entanglement by being unitarily equivalent to a Bell pair. In this case, the natural measure is the entanglement entropy, which is vanishing for the first class and maximal for the second one.

\subsubsection{3-qubits}
The situation becomes slightly richer for $3$-qubits, as discussed in \cite{brayvi2006ghz}. There are five entanglement classes:
\begin{equation}
    \begin{split}
        \left[\mathrm{STAB}_{2,3}\right]_\mathrm{LC} = \Bigg\{&\left[|000\rangle\right]_\mathrm{LC}\quad ,  \\ & \left[\frac{1}{\sqrt 2}(|000\rangle + |110\rangle)\right]_\mathrm{LC}\quad, \\ & \left[\frac{1}{\sqrt 2}(|000\rangle + |101\rangle)\right]_\mathrm{LC}\quad, \\ & \left[\frac{1}{\sqrt 2}(|000\rangle + |011\rangle)\right]_\mathrm{LC}\quad,\\ & \left[\frac{1}{\sqrt 2}(|000\rangle + |111\rangle)\right]_\mathrm{LC}\Bigg\}\;\text{ }.
    \end{split}
\end{equation}
The first class corresponds to separable states, the last to the $3$-qubit GHZ state, and the other three all correspond to a bi-separable state.

\section{Geometric Analysis and Witnesses}\label{ref:geometric_analysis}

Building on the definitions and background provided in the previous section, we now present our main analytical contributions. We derive explicit relations between the trace distance and the violation of facet inequalities for low-dimensional systems, identify the states that maximize these violations, and explore connections to Bell-like inequalities.

\subsection{Analytical Expression for NTD}

Here, we present our main contribution relating the violation of a facet inequality with the NTD, obtaining an analytical expression for the latter. 
\subsubsection{Analytical expression for the NTD of one qubit.} \label{subsec:analytical_expression_TD}
In this section, we obtain explicitly an analytical expression for our measure of non-stabilizerness defined through the trace distance in Eq.~(\ref{eq: Magic_TD}). The program~(\ref{eq:SDP_MDistance-2}) uses the $V$-representation of the polytope, but here we reformulate it in its $H$-representation because this will allow us to relate the distance to the facet analytically with the hyperplanes parallel to the facet.
The SDP for the non-stabilizerness by trace distance of one $d$-dimensional system reads as :
\begin{align}
\mathcal{M}(\rho) = &\min_\sigma \frac {1} {2} \| \rho-\sigma \|_1\text{ }, \label{eq:SDP_MDistance-1}\\
\text{subject to} \nonumber\\
 &\text{Tr}(\sigma A^{\bf{q}}) \geq 0\text{ },\text{ } \forall {\bf q} \in \mathbb{Z}_2^{3}\text{ }, \nonumber\\
 &\text{Tr}(\sigma)=1\nonumber.
\end{align}
where $A^{{\bf q}}$ are the operators defining the facets of the polytope as defined in Eq.~(\ref{eq: H-rep}). We emphasize that solving this problem requires casting it as a Semidefinite Program. The trace distance objective is reformulated into a linear objective with additional linear matrix inequality constraints. While the constraint $\sigma \ge 0$ is mathematically redundant for any $\sigma$ within the stabilizer polytope, the SDP formulation is needed by the trace distance objective, which is the standard method for solving such problems numerically.

What we will do in the following is to consider all the states $\rho$ such that  $\text{Tr}(\rho A^{{\bf q'}}) = c$ for some ${\bf q'}$, that is, we will look at the ``iso-non-stabilizerness'' surfaces of states.
We are interested in the range $c \in [- v_{\text{max.}},0]$ where $- v_{\text{max.}} <0$ corresponds to the maximal violation possible. Since for $1$ qubit all facets are equivalent to $\text{Tr}(\rho A^{{\bf 0}}) =0$, it is enough to consider only the states $\rho$ such that $\text{Tr}(\rho A^{{\bf 0}}) = c$. Then the result obtained for this facet can be generalized for all other facets by symmetry. 

We consider then, the distance from the set of states belonging to the hyperplane $\text{Tr}(\rho A^{{\bf 0}})=c$ to the polytope $P^{\text{STAB}}_{2,1}$. This corresponds to the following minimization problem:
\begin{align}
\mathcal{M}_{lb} = \min_{\sigma,\rho}& \frac {1} {2} \| \rho-\sigma \|_1\text{ }, \label{eq:SDP_MRhoSigmaDIst}\\
  &\text{Tr}(\sigma A^{{\bf q}}) \geq 0\text{ },\text{ } \forall {\bf q} \in \mathbb{Z}_2^{3} \nonumber\\
  &\text{Tr}(\rho A^{{\bf 0}}) = c\text{ },  \nonumber\\
  &\text{Tr}(\rho ) = 1\text{ },  \nonumber \\
  &\rho  \succeq 0\nonumber \text{ }.
\end{align}

This optimization gives a tight lower bound for the trace distance to the stabilizer polytope for all the states that violate the inequality  $\text{Tr}(\sigma A^{{\bf q}}) \geq 0$ by an amount of $c$. One can think of this quantity as a coarse-grained version of $\mathcal{M}$, where one disregards the finiteness of the stabilizer polytope. In general, $\mathcal{M}_{lb}(\rho)=\mathcal{M}(\rho)$ whenever the region of the hyperplane containing the state $\rho$, when projected to the hyperplane containing the facet of the polytope, is actually contained in the facet. If $\rho$ is in a region whose projection is outside of the facet of the polytope $\mathcal{M}_{lb}(\rho)\geq\mathcal{M}(\rho)$ in general.  

By using the usual trick for the norm minimization where $\min \| A\|_1$ can be expressed as $\min \text{Tr}(t)$, such that $t\succeq A\succeq-t$, and by  defining the block matrix:
\begin{equation}
X = \begin{bmatrix}
\rho & 0& 0 \\
0 & t & 0 \\
0 & 0 & \sigma \\
\end{bmatrix},
\end{equation}
one can reduce the problem to a one-variable minimization:
\begin{align}
\mathcal{M}_{lb}= \min_{X}& \frac{1}{2}\text{Tr}(E_5X) \text{ }, \label{eq:SDP_MX}\\
    &\text{Tr}(XE_9) = 1\text{ }, \nonumber\\
  &\text{Tr}(XE_1) = 1\text{ },  \nonumber\\
  &\text{Tr}(XA_1^0) = c\text{ },  \nonumber \\
  &E_1(-X+E_2XE_4+E_3XE_7)  \succeq 0\text{ },  \nonumber \\
   &E_1(X+E_2XE_4-E_3XE_7)  \succeq 0\text{ },  \nonumber \\
  &\text{Tr}(XA_9^i)\geq0\text{ },\text{ } \forall  i= 1,...,8\text{ },\nonumber\\
  &X \succeq 0\nonumber\text{ }, 
\end{align}
where $E_i$ is the block matrix with identity in the $i$-th position and $A_9^i$, is the block matrix with the $ A^{{\bf i}}$ operator in the ninth entry. The process is similar for $A_1^0$. From this, it is straightforward to construct the Lagrangian: 
\begin{align}
    L &= \frac{1}{2}\text{Tr}(E_5X)- y_1[\text{Tr}(XE_9)-1]  \nonumber \\
    & - y_2[\text{Tr}(XE_1)-1]- y_3[c-\text{Tr}(XA_1^0)] \nonumber \\
    &- \text{Tr}\left[Z_1E_1(-X+E_2XE_4+E_3XE_7)\right] \nonumber \\
    &- \text{Tr}\left[Z_2E_1(X+E_2XE_4-E_3XE_7)\right] \nonumber \\
    &- \sum_{i=0}^7z_{3+i}\text{Tr}(XA_q^i)-\text{Tr}(Z_1X)\text{ },
\end{align}
where $y_i,z_i$ and $Z_i$ are the associated dual variables corresponding to the dual program of the SDP~(\ref{eq:SDP_MX}). All $Z_i$ are positive semi-definite, and $z_i\geq0$. This Lagrangian defines the objective function of the dual problem:
\begin{equation}
    -y_1-y_2-cy_3\text{ }. 
\end{equation}
Now, we recall the strong duality theorem:  whenever the primal is strictly feasible (i.e., if we replace all the inequality constraints by strict inequality constraints and the feasible set is still not empty) and bounded, then the value of the dual and of the primal coincide \cite{skrzypczyk2023semidefinite}. Since this is the case for our problem, as there always exist positive definite states $\rho$ and $\sigma$ satisfying  $\text{Tr}(\rho A^{{\bf 0}}) = c$ and  $\text{Tr}(\sigma A^{{\bf 0}}) > c$ leaving  the  norm minimization  bounded for $c \in [- v_{\text{max.}},0]$, then we can see that:
\begin{equation}
    \mathcal{M}_{lb}=\max_{y_i,Z_j,z_k} -y_1-y_2-cy_3\text{ },
\end{equation}
which implies that $\mathcal{M}_{lb}$ is a linear function of $c=\text{Tr}(\rho A^{{\bf 0}})$. Interestingly, when looking at the numerically optimized values, we find that for the interval of interest $c \in [-  v_{\text{max.}},0]$, $-y_1-y_2=0$ and $y_3=\sqrt{3}/3$, which gives the result:

\begin{equation}
    \mathcal{M}_{lb}(\rho)=-\frac{\sqrt{3}}{3}\text{Tr}(\rho A^{{\bf 0}})\text{ }.
\end{equation}
Since all ${\bf q}$ facets are equivalent, for an arbitrary state, one can define the following quantifier as follows:
\begin{equation}
    \mathcal{M}_{lb}(\rho)=-\frac{\sqrt{3}}{3}\min[0,\{\text{Tr}(\rho A^{{\bf q}})\}_{{\bf q}}]\text{ },
\end{equation}
where the minimum is taken over the value of all facets (given $\rho$) and $0$. Since the expectation values $\langle X\rangle$, $\langle Y\rangle$, and $\langle Z\rangle$ correspond to the Bloch vector components $r_x$, $r_y$, and $r_z$, the above expressions provide a geometric interpretation of the NTD as the Euclidean distance from $\Vec{r}$ to the facet of the stabilizer octahedron. This contrasts with the RoM, which uses the $\ell_1$-norm \cite{seddon2021quantifying}. This geometric interpretation of the NTD is easily seen from:
\begin{equation}
    \begin{split}
        \mathcal{M}_{lb}(\rho) &= -\frac{\sqrt{3}}{3}\text{Tr}(\rho A^{{\bf 0}}) \\
        &= -\frac{1}{\sqrt{3}}\text{Tr}\left[\rho (1 + X + Y + Z)\right] \\
        &= -\frac{1}{\sqrt{3}}\left[\text{Tr}(\rho) + \text{Tr}(\rho X) + \text{Tr}(\rho Y) + \text{Tr}(\rho Z)\right] \\
        &= -\frac{1}{\sqrt{3}}[1 + r_x + r_y + r_z].
    \end{split}
\end{equation}

which is the Euclidean distance from the point $\Vec{r}$ to the octahedron plane with a minus sign.

\begin{figure}[ht!]
    \centering
    \includegraphics[width=\columnwidth]{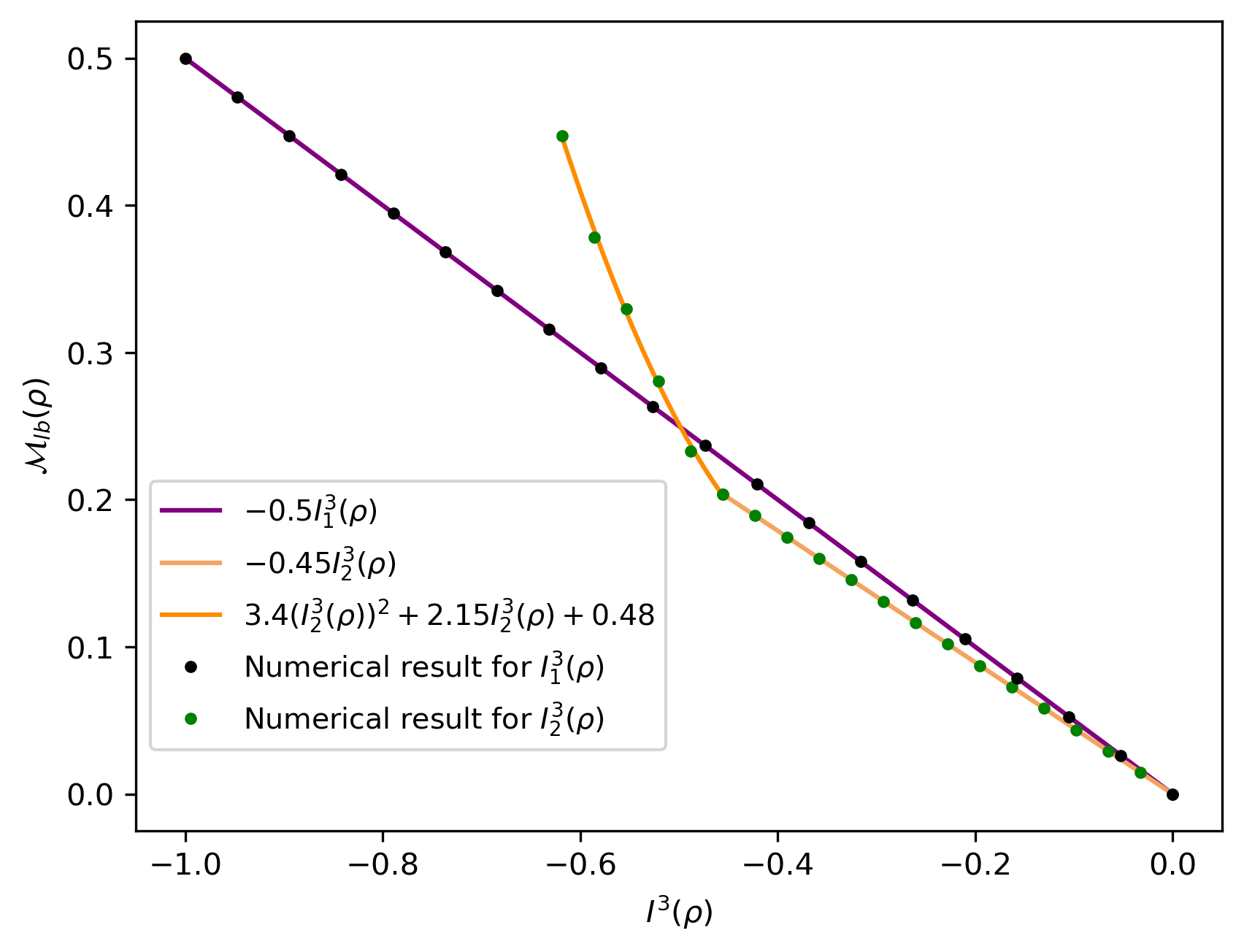}
    \caption{\raggedright Tight lower bound of $\mathcal{M}(\rho)$ restricted to  a value  $I^3_i(\rho)$ for $i=1,2$. In this figure, we show how the numerical approach allows to get an analytical expression for the SDP [see Eq. (\ref{eq: Magic_TD})]. The dots are the numerical result of the SDP by restricting the states to satisfy a certain value of the witnesses $I^3$ given by Eq. (\ref{eq: Ineq1_qtrits}) (black dots) and Eq.~(\ref{eq: Ineq2_qtrits}) (green dots). We show the corresponding fitting in the legend, giving the analytical expression for the trace distance in terms of $I^3$.}
    \label{fig: Qutrits_Analytical_TD}
\end{figure}

\subsubsection{NTD of one qutrit.}
Now we relate the violation of the witnesses Eq.~(\ref{eq: Ineq1_qtrits}) and Eq.~(\ref{eq: Ineq2_qtrits}) to the NTD for a qutrit.  The dimension of the problem is large enough to make the analytical approach intractable. Therefore, we use a numerical approach in which we perform an analogous optimization to the SDP~(\ref{eq:SDP_MRhoSigmaDIst}), but now restricting $\sigma\in P_{3,1}^{STAB}$ and considering the violation $I^3_i = c_i$  for $c_{i} \in [-mV(I^3_{i}), 0]$ where $-mV(I^3_{i})$ is the maximum violation of $I^3_i$, $i=1,2$. In Fig.~\ref{fig: Qutrits_Analytical_TD}, we show a plot with the results. 

The result implies that  for Eq.~(\ref{eq: Ineq1_qtrits})  one can write $\mathcal{M}(\rho)$ as:
 \begin{equation}
    \mathcal{M}_{lb}\left[I^3_1(\rho)\right]=-\frac{1}{2}I^3_1(\rho)\text{ },
\end{equation}
and, for Eq.~(\ref{eq: Ineq2_qtrits}):
\begin{equation}
     \mathcal{M}_{lb}\left[I^3_2(\rho)\right]= 
        \begin{cases}
           -\frac{1}{\sqrt{5}}I^3_2(\rho)\text{ },& \text{if } -i_0\leq I^3_2(\rho) \leq 0 \\
           P\left[I^3_2(\rho)\right]\text{ }, & \text{if } -i_1\leq I^3_2(\rho) <-i_0
        \end{cases}\text{ },
\end{equation}
 where $i_0 = 0.4554$, $i_1 = 0.618$ and $P\left[I^3_2(\rho)\right] = 3.43\left[I^3_2(\rho)\right]^2 + 2.19\left[I^3_2(\rho)\right]+ 0.49$.
 Considering all the other facets of the polytope that can be obtained through conjugation of Clifford unitaries, we can write the non-stabilizerness for a general state $\rho$ as:
\begin{equation}
    \mathcal{M}_{lb}(\rho)=\min[0,\mathcal{M}(\{I^3_1(\rho)\}_{\mathcal{C}}), \mathcal{M}(\{I^3_2(\rho)\}_{\mathcal{C}})]\text{ },
\end{equation}
where $\{I^3(\rho)\}_\mathcal{C}$ denotes  all the  Clifford symmetries of $I^3(\rho)$.

\subsubsection{Other insights for 2 qbits}

The previous technique extends to the two-qubit case, which features $8$ equivalence classes of facets. For brevity, we omit the analysis as the method is identical to the one previously described.  Instead, we present the following results.

To assess the detection capability of these facets, we sampled $10^6$ mixed states from the Hilbert-Schmidt ensemble using the QuTiP \cite{lambert2024qutip} \href{https://qutip.org/docs/4.0.2/modules/qutip/random_objects.html#rand_dm}{\texttt{rand\_dm}} function, and we plot a histogram showing the percentage of states that violate each facet (or one of its symmetries). We show the results in Fig.~\ref{fig: violation_of_ineqs}. When sampling pure two-qubit states from a Haar ensemble, we find that all of them violate at least one of the symmetries of all $8$ classes of inequalities. From a sample of $10^5$, we also calculated the average number $\overline n_f$ of facets that a $2$ qubit state violates. Pure states violate, on average, $\overline n_f \approx 12677$ facets, while for mixed states this number decreases to $\overline n_f \approx 1228$.

\begin{figure}[ht!]
    \centering
    \includegraphics[width=\columnwidth]{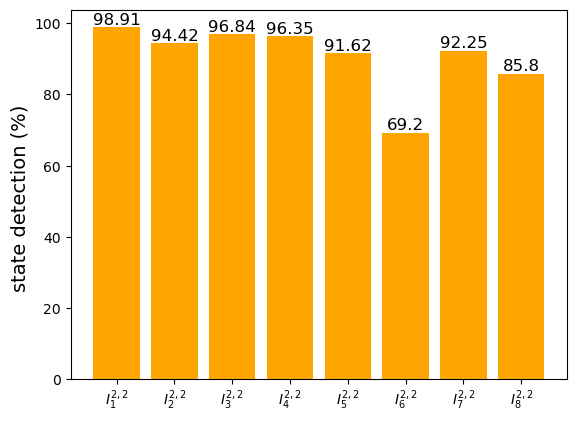}
    \caption{\raggedright Detection capability of the facets of the 2-qubit polytope. The plot shows the percentage of states that violate each facet out of $10^6$ states sampled randomly from the Hilbert-Schmidt ensemble.}
    \label{fig: violation_of_ineqs}
\end{figure}

We also derived the states that maximally violate these facets. Optimizing over all states, we found that the maximal violation of the inequalities $I^{2,2}_r$ is achieved by the states shown in Table~\ref{2Qbits_stab_states}. The optimization can be stated as an SDP, or simply by a brute force optimization over a two-qubit state parameterized as:
\begin{align}\label{psigeral}
        \ket{\{\theta_j,\phi_j\}_{j=1}^3} &= \cos(\theta_1)\ket{00}  \nonumber\\
        &+\exp(i\phi_1)\sin(\theta_1)\cos(\theta_2)\ket{01} \nonumber \\
        &+ \exp(i\phi_2)\sin(\theta_1)\sin(\theta_2)\cos(\theta_3)\ket{10} \nonumber\\
        &+ \exp(i\phi_3)\sin(\theta_1)\sin(\theta_2)\sin(\theta_3)\ket{11},
\end{align}
for $\theta_j \in [0,\pi/2]$ and $\phi_j\in [0,2\pi]$ with $j=1,2,3$. By choosing the latter option, we can obtain the coefficients of the states shown in the first four columns of Table~\ref{2Qbits_stab_states}. For example, the state violating maximally $I^{2,2}_1$ is:
\begin{equation}
    \ket{\psi^{2,2}_1} = \alpha\ket{00} - \beta^* \ket{01} - \beta\ket{10} + i\alpha\ket{11}\text{ },
\end{equation}
where $\alpha =0.628$ and $\beta = 0.23(1+ i)$. The fifth and sixth columns show the maximal violation of the corresponding inequalities and the non-stabilizerness of the state calculated via the NTD measure. It is worth noting that the state $|T\rangle^{\otimes 2}$ (which we identify in Sec. \ref{sec:distribution} as a candidate for maximal NTD via random sampling) has a value of $\mathcal{M}(|T\rangle^{\otimes 2})=0.378$, which is higher than any of the listed states. By evaluating $|T\rangle^{\otimes 2}$ in all inequalities (all symmetries of ineqs. in Table \ref{tab:2Qbits_Stab_ineqs}), we found that indeed the violation of this state is lower than all the values presented in Table \ref{tab:2Qbits_Stab_ineqs}. This means that $|T\rangle^{\otimes 2}$ is located close to a hyperplane containing some facet of the polytope but far away from the polytope. In Fig.~\ref{fig: T2_avocado}, we show an example of how this could happen in $2$ dimensions. 
\begin{figure}[ht!]
    \centering
    \includegraphics[width=0.8\columnwidth]{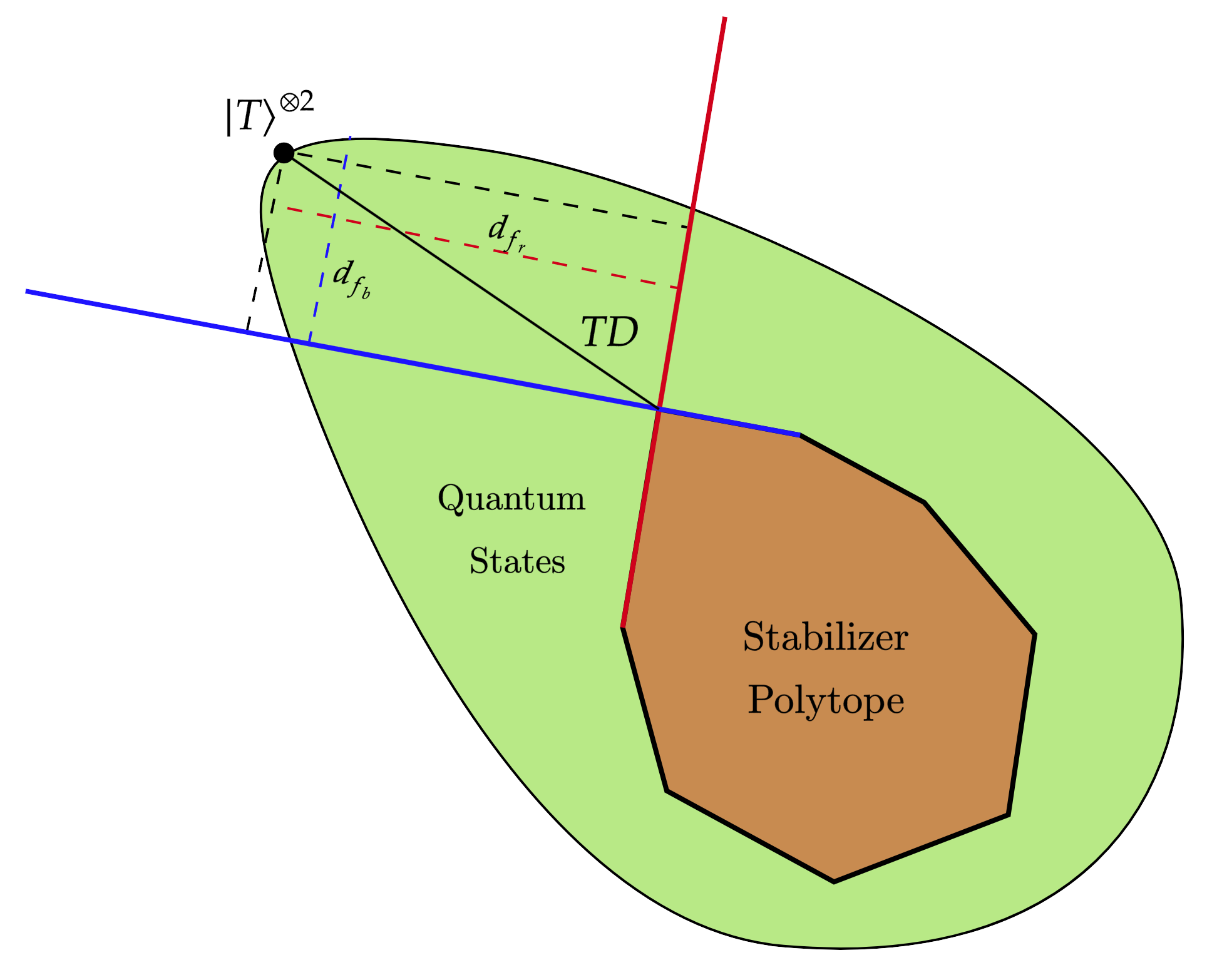}
    \caption{\raggedright An avocado-like geometry, representing a far-away non-stabilizer state close to the facets of the stabilizer polytope. In this pictorial example, we show how  $|T\rangle^{\otimes 2}$ can be the furthest away state from the stabilizer polytope as quantified by the trace distance NTD, but there exist points on the border of the set of quantum states whose distances to the closest blue facet ($d_{f_b}$) and red facet ($d_{f_r}$), are larger than the distance of  $|T\rangle^{\otimes 2}$ to these facets, compare the projections given by the black dashed lines to the blue and red dashed lines. }
    \label{fig: T2_avocado}
\end{figure}
\begin{table*}[ht]
\centering
\small
\renewcommand{\arraystretch}{1.5} 
\resizebox{0.7\textwidth}{!}{
\begin{tabular}{|c|c|c|c|c|c|c|}
\hline
& $\ket{00}$ & $\ket{01}$ & $\ket{10}$ & $\ket{11}$ &  $\text{min}(I^{2,2}_r)$ & $\mathcal{M}$  \\
\hline
$\ket{\psi^{2,2}_1}$ & $0.628 $& $-0.23 + i0.23$ & $-0.23 - i0.23$ &  $i0.628$ & $-1.4641$ & $0.2113$\\
$\ket{\psi^{2,2}_2}$ & $0.773$& $-0.394 + i0.273$ &$ -0.083 + i0.241$ & $-0.255 + i0.248$ & $-1.2361 $  & $0.3577$ \\
$\ket{\psi^{2,2}_3}$ & $0.817$ & $-0.299 - i0.299$ & $-0.335 - i0.09$ & $-0.09 - i0.156$ & $-1.6397$ & $0.3376$\\
$\ket{\psi^{2,2}_4}$ & $0.648$ & $0.185 - i0.535$ & $-0.262 + i0.108$ &$ 0.381 - i0.185 $ & $-1.6171$ & $0.2998$\\
$\ket{\psi^{2,2}_5}$ & $0.765$& $-0.383 + i0.247$ &$ 0.068 + i0.315$ & $-0.315 + i0.068$ & $-1.2915$ & $0.3475$ \\
$\ket{\psi^{2,2}_6}$ & $0.459$ & $0.579 - i0.164$ & $-i0.256$ & $-0.579 + i0.164$ &$-1.8042$ & $0.3304$\\
$\ket{\psi^{2,2}_7}$ & $0.544$ & $-0.272 - i0.035$ & $-0.391 + i0.119 $& $-0.663 + i0.153$ & $-1.746$&$0.2334$\\
$\ket{\psi^{2,2}_8}$ & $0.372 $ & $-0.715-i0.399$ & $- 0.113+i0.203$ & $-0.195 - i0.316$ & $-1.4875$ & $0.3304$\\
\hline
\end{tabular}
}
\caption{\raggedright States $\ket{\psi^{2,2}_r}$ that maximally violate the inequalities $I^{2,2}_r$.  We also show the achieved violation $\text{min}(I^{2,2}_r)$ of the corresponding inequality and the Non-stabilizerness by Trace Distance, $\mathcal{M}$, of the state.}\label{2Qbits_stab_states}
\end{table*}

\subsection{Bell-like Inequalities}

To further explore the stabilizer set, we revisit the results from Ref.~\cite{howard2012nonlocality}, establishing a connection between non-stabilizerness detection and the violation of the CHSH inequality~\cite{clauser1969proposed}. Traditionally, a Bell inequality violation serves as a Device-Independent (DI) proof $-$ agnostic to details of state preparation and measurement $-$ that the tested state is entangled. Recent work \cite{macedo2025witnessing} has shown that Bell inequalities can also be repurposed as DI tests of non-stabilizerness.

Further insights emerge when moving to a Device-Dependent (DD) setting, assuming the measurements in the Bell test are fixed as Pauli measurements. In this case, a DD version of the CHSH inequality takes the following form \cite{howard2012nonlocality}:
\begin{align} \label{CHSH} \braket{X\otimes X} + \braket{X\otimes Y} + \braket{Y\otimes X}- \braket{Y\otimes Y} \leq 2\text{ }, \end{align}
which was shown to be a valid constraint on the stabilizer polytope. A natural question is then whether this DD version of the CHSH inequality defines a facet of the stabilizer polytope when restricted to the four observables $X\otimes X$, $X\otimes Y$, $Y\otimes X$, $Y\otimes Y$. This projection of the stabilizer polytope is governed by a single class of inequalities, including
\begin{equation}\label{CHSH2}
    \begin{split}
        & \braket{X\otimes X} + \braket{X\otimes Y} \leq 1, \\
        & \braket{Y\otimes X} - \braket{Y\otimes Y} \leq 1,
    \end{split}
\end{equation}

which sums precisely to the CHSH inequality. Thus, while the DI CHSH inequality forms a facet of the local polytope in Bell’s theorem, its DD counterpart used in the characterization of non-stabilizerness does not. Another intriguing result is that even Bell-type inequalities involving only two terms can act as effective non-stabilizerness witnesses. Interestingly, when these so-called DD inequalities are maximized for the general state given in Eq. (\ref{psigeral}), both exhibit a violation of $\sqrt{2}$. For instance, for the first inequality, the optimal parameters for the state achieving maximal violation are given by $\theta_1 \approx -1.57$, $\theta_2 \approx 0.78$, $\theta_3 = 7.48 \times 10^{-9}$, $\phi_1 \approx -0.40$, $\phi_2 \approx -1.19$, and $\phi_3 \approx -0.17$.

For the 3-qubit case, a full facet enumeration is intractable. Nevertheless, we can still gain insight by analyzing lower-dimensional projections. The first such projection we consider is the three-body polytope projection, where we restrict attention to vertices with only full correlator entries given by:
\begin{equation}
    \langle p_1\otimes p_2\otimes p_3\rangle\text{ }, 
\end{equation}
where $p_i \in \{X,Y,Z\}$. This reduces the full-dimensional polytope $P_{2,3}^{STAB}$ to a $27$-dimensional projection and decreases the number of vertices to $918$. With this reduction, we are now able to identify a significantly larger number of facets—unlike in the full polytope case, where only 2 facets were found after $120$ hours of computation using mplrs \cite{avis2001lrs}. In contrast, for the reduced projection, halting after the same $120$ hours, we identified approximately $40,000,000$ facets.

To show the relevance of these inequalities, we tested whether the W-state \cite{dur2000three}:
\begin{equation} \label{eq:W-state}
    \ket{W} = \frac{1}{\sqrt{3}}\left(\ket{001}+\ket{010}+\ket{100}\right)\text{ },
\end{equation}
and the Hoggar state (\ref{eq:hoggar}) would violate such inequalities. 
For the $W$ state, among all the facets obtained, the largest violation observed (equal to $-714.66$) corresponds to the following facet of the projected polytope:
\begin{align}
    (&-141, -103, -153, -359, -141,  -71,  -10,  -72,-146,\nonumber\\
    &-79,  141,   67,  101,  -99,   17,   78, -292,366,
     -157,  -95, 151,\nonumber \\
     & -225,  -45,  -89, -170,  232,
        288)\cdot {\bf P}_3 \geq -608\text{ },
\end{align}
where ${\bf P}_3$ stands for the three-body correlators vector associated with the state under test:
\begin{equation}
    {\bf P}_3 = (\langle X\otimes X\otimes X \rangle, \langle X\otimes  X\otimes Y \rangle,\dots, \langle Z\otimes Z\otimes Z\rangle)\text{ }.
\end{equation}
For the Hoggar state, we found it has the highest violation (of  $-1518.66$) for an  inequality given by
\begin{align}
    (  &278, -348, -500, -330, -116,   42, -204,  406,-60,\nonumber \\
      &  317,   15,  143,  687, -241,  399,  567, -229,303, -27, -419, \nonumber\\
      &   339, -181,  507, -433, -271,  609, 89)\cdot{\bf P}_3 \geq -1236\text{ }.
\end{align}

Similarly to what has been done in the case of 2-qubits, we have also obtained the projected stabilizer polytope considering only two-body expectation values corresponding to the Pauli operators $X$ and $Y$ among all pairs of qubits, that is, correlators of the form: $\langle X \otimes X \otimes I \rangle$, $\langle X \otimes I \otimes X \rangle$, $\ldots$, $\langle I \otimes Y \otimes Y \rangle$. As a result, the $1080$ vertices of the full stabilizer polytope are reduced to only $80$. After transitioning to the H-representation, $9984$ facets were obtained. After evaluating the inequalities for the W-state, we found that this state violates 840 of them, while the Hoggar state defined in Eq. (\ref{eq:hoggar}) violates 122. One such inequality is given by:
\begin{align}\label{chsh3q}
    &2\braket{X_1 X_3} - \braket{X_1 X_2}
     +3\braket{X_2   X_3} -3\braket{X_1 Y_2} \nonumber \\
    & +\braket{Y_1   X_2} - 2\braket{X_1 Y_3} -3 \braket{X_2  Y_3 } + \braket{Y_2 X_3} \nonumber \\
    &+3 \braket{Y_1 Y_2} +2 \braket{Y_1 Y_3} - \braket{Y_2 Y_3}\leq 4 \text{ },
\end{align}
where we have used the short-hand notation  $\braket{I \otimes X \otimes X}=\braket{X_2 X_3}$ and similarly for the other terms. Inequality (\ref{chsh3q}) is violated by the Hoggar state with a value of $\approx 4.66$ and by the W-state at $\approx 5.33$.

Inspired by the results in \cite{howard2012nonlocality}, showing that the CHSH inequality is a valid stabilizer witness if one assumes the measurements to be Pauli operators, one can wonder whether other paradigmatic Bell inequalities will show a similar behaviour. Unfortunately, as we discuss next, this is not the case. Take, for instance, the multipartite Mermin inequality \cite{mermin1990extreme}, which for 3 qubits and Pauli measurements can be defined as:
\begin{equation}
\braket{XXX}-\braket{XYY}-\braket{YXY}-\braket{YYX} \leq 2\text{ }.
\end{equation}
This inequality achieves the maximum algebraic violation for a GHZ state $\ket{\text{GHZ}}=(1/\sqrt{2})(\ket{000}+\ket{111})$ that is a stabilizer state, hence the Mermin inequality cannot be a valid witness (a similar behavior being valid for any number of qubits).

\section{Statistical Distribution and Noise Resilience}\label{sec:distribution}

Complementing the previous geometric analysis, this section presents numerical results on the distribution of non-stabilizerness and its robustness against noise.

\subsection{Non-stabilizer distribution}

It is known that generic quantum states are highly entangled \cite{hayden2006aspects}, a phenomenon known as concentration of measure, where if we sort a quantum state at random, with high probability it will be almost maximally entangled with respect to any bipartition of its systems. Interestingly, similar results also hold for non-stabilizerness \cite{leone2022stabilizer, liu2022many,chen2024magic,turkeshi2025magic}. While these results hold asymptotically, much less is known about the distribution of non-stabilizerness for a small number of qubits $-$ a question we address next by randomly sampling quantum states of $1$, $2$, and $3$ qubits.

We have employed Qiskit's \cite{qiskit2024} \href{https://docs.quantum.ibm.com/api/qiskit/0.19/qiskit.quantum_info.random_statevector}{\texttt{random\_statevector}} function to sample random pure states. This function samples states from a uniform distribution based on Haar measure. For mixed states, we utilized Qiskit's \href{https://docs.quantum.ibm.com/api/qiskit/0.35/qiskit.quantum_info.random_density_matrix}{\texttt{random\_density\_matrix}} function, which is based on the Hilbert-Schmidt measure \cite{zyczkowski2001induced}.
For each of the cases, we sampled $360,000$ random pure and mixed states, which, as we will see, already provide a reasonable sample to understand the behavior of the non-stabilizerness distribution qualitatively.

The results for all cases analyzed can be seen in Fig.~\ref{fig:mtd_hist_all}. For the mixed single-qubit case, we have a peak at zero non-stabilizerness. The highest value of non-stabilizerness was $\mathcal{M}(\rho) = 0.211$ for the state $\ket{T}$ with Bloch vector $(r_x,r_y,r_z) =(1,1,1)/\sqrt{3}$ \cite{bravyi2005universal}. The highest value of non-stabilizerness for 2 qubits is given by $\mathcal{M}(\rho) = 0.378$, corresponding to the state $\ket{T}^{\otimes 2}$, that is, to a separable two-qubit state. In turn, for 3 qubits, the highest value of non-stabilizerness obtained was $\mathcal{M}(\rho) = 0.583$, corresponding to the Hoggar state given by \cite{howard2017application, hoggar199864}:
\begin{align}
    \label{eq:hoggar}
    \ket{\text{Hoggar}} = \frac{1}{\sqrt{6}}\Big[(1+i)\ket{000} - &\ket{010} + \ket{011} \nonumber\\
    -i &\ket{100} + \ket{101}\Big]\text{ },
\end{align}
That is different from the 2-qubit case, which is an entangled state. For comparison, the NTD for $\ket{T}^{\otimes 3}$ is given by $\mathcal{M(\rho)}=0.509$.

This simple analysis already allows us to qualitatively understand how the distribution of non-stabilizerness, as quantified by the NTD, evolves with an increasing number of qubits. The distribution exhibits a gap between the mixed and pure cases. That is, as expected, the distribution of non-stabilizerness for mixed states tends to concentrate around smaller values.  Moreover, as the number of qubits increases, the peak of the distribution shifts toward higher non-stabilizerness values while its width decreases, reflecting a scaling behavior consistent with the asymptotic concentration described in the literature \cite{liu2022many}.


We emphasize that our random state sampling identified the $\ket{T}$, $\ket{T}^{\otimes 2}$, and $\ket{\text{Hoggar}}$ states as yielding the highest NTD values, which does not imply they are the states with the highest NTD. However, while this numerical search is not exhaustive, it provides strong candidates for states with maximal NTD. A comparison with states known to maximize the Robustness of Magic~\cite{howard2017application} reveals that for one and three qubits, the states with the highest RoM also exhibit the highest NTD found in our search. In contrast, for the two-qubit system, the state with the highest RoM has an NTD of $0.330$, which is lower than that of the $\ket{T}^{\otimes 2}$ state.

\subsection{Comparison of Measures}

To validate the NTD as a reliable resource quantifier, we benchmark it against the RoM and SRE using the sampled datasets.

\begin{figure}[ht!]
    \centering
    \includegraphics[width=\columnwidth]{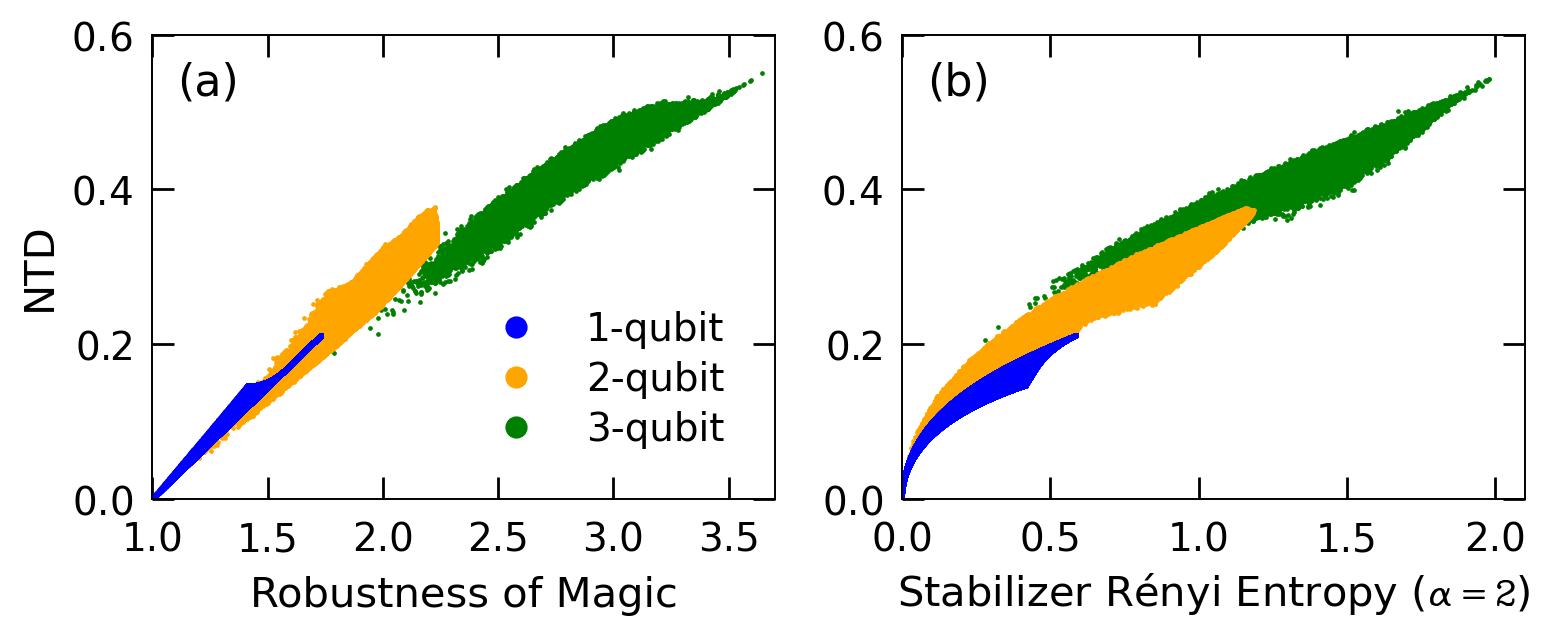}
\caption{\raggedright Non-stabilizerness by Trace Distance versus \textbf{(a)} Robustness of Magic and \textbf{(b)} Stabilizer Rényi Entropy ($\alpha=2$) for randomly sampled 1-, 2-, and 3-qubit pure states. A total of 360,000 pure states were sampled for each qubit.}

    \label{fig:MTD_RoM_Renyi_Pure.jpg}
\end{figure}
Using the same pure-state dataset as above, we also compared the NTD with two other non-stabilizerness measures: RoM and the stabilizer Rényi Entropy. As shown in Fig. \ref{fig:MTD_RoM_Renyi_Pure.jpg}, the behavior of the NTD in relation to RoM and the stabilizer Rényi Entropy is very similar. There is no one-to-one correspondence between these measures, as each NTD value corresponds to a continuous but relatively narrow range of values for the other two measures. Nevertheless, as expected, states classified as highly (barely) non-stabilizers according to one measure tend to also exhibit high(low) non-stabilizerness under the other measures.

\begin{figure}[ht!]
    \centering
    \includegraphics[width=\columnwidth]{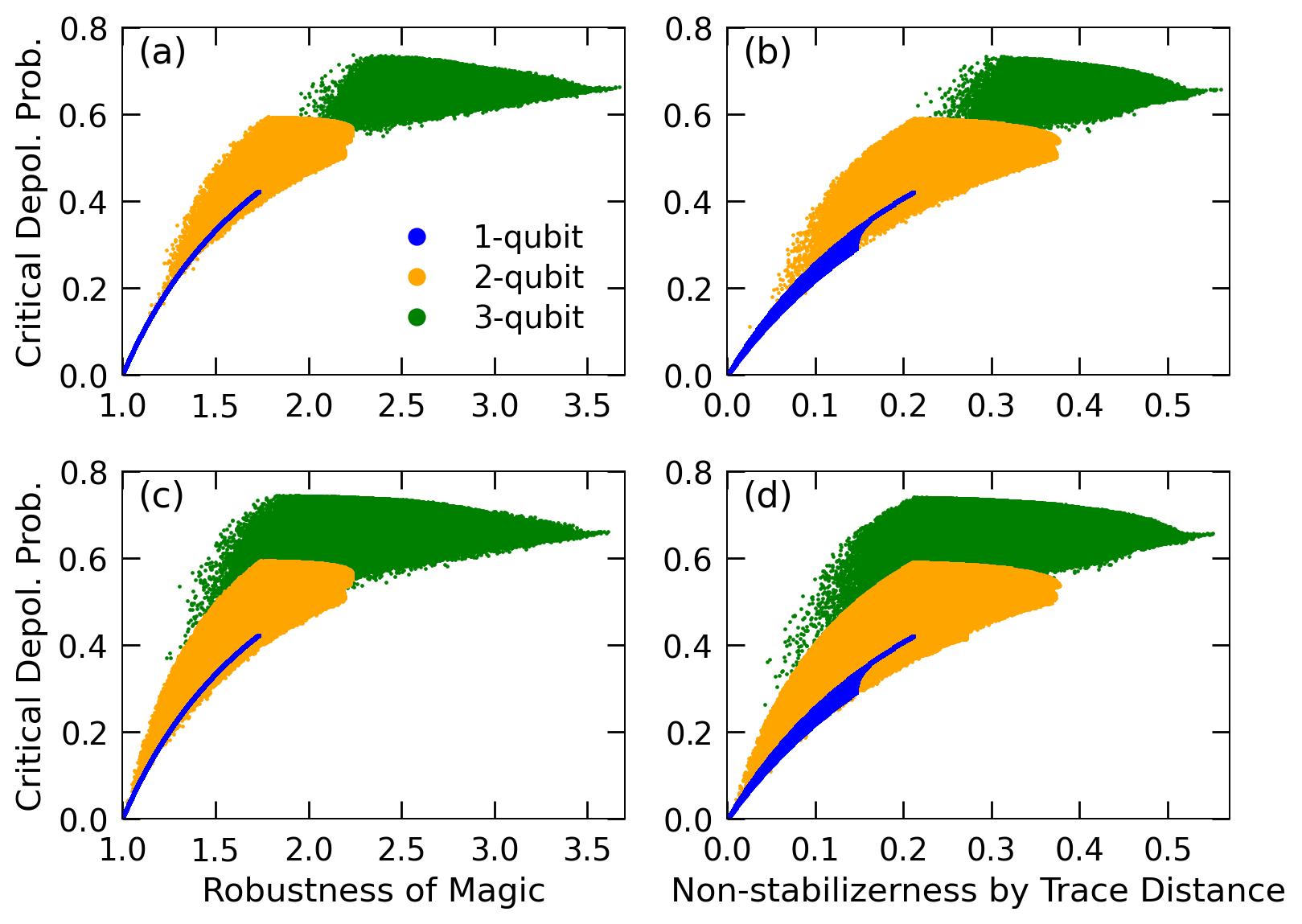}
  \caption{\raggedright
    Critical depolarization probability for pure states versus: \textbf{(a)} RoM and \textbf{(b)} NTD, using uniformly random quantum states;   \textbf{(c)} RoM and \textbf{(d)} NTD, using biased random quantum states.  
    For each case ($1$, $2$ or $3$ qubits), 1,900,000 states were sampled.  
}
    \label{fig:CriticalDepolProbvsMagic}
\end{figure}

\subsection{Robustness to Noise}\label{subsec:robustness_depolarization}

 To understand the noise resilience of the non-stabilizerness of a quantum state, we investigate the depolarization probability at which a pure state loses its non-stabilizerness, i.e., it enters the stabilizer polytope. This analysis is motivated by the operational relevance of the NTD: any state within the stabilizer polytope is efficiently simulable using the Aaronson-Gottesman algorithm \cite{aaronson2004improved}.

More precisely, we consider a global depolarization scenario where, given a randomly sampled state vector $\psi$  with density matrix $\rho=\lvert \psi \rangle \langle \psi \rvert$, we consider its depolarized version $(1-p)\rho+p\openone/d$ (where $d$ is the Hilbert space dimension of the state under scrutiny). We also consider a local depolarization scenario, where the depolarization noise is applied to each qubit individually, assuming that the depolarization parameter $p$ is the same for each noise channel.

We analyze the relationship between the RoM and depolarization, as well as the relationship between NTD and depolarization (see Fig. \ref{fig:CriticalDepolProbvsMagic}).

For the 1-qubit case, we observe that the threshold depolarization probability $p$ above which all states enter the stabilizer polytope is given by $p\approx 0.42$. For $2$ and $3$ qubits, this value increases to $p\approx 0.59$ and $p\approx 0.74$, respectively. The thresholds obtained using RoM and NTD are numerically very close, indicating, as expected, consistency between different measures to quantify the threshold depolarization upon which the set of quantum states becomes classically simulable.

As evident from the histograms in Fig. \ref{fig:mtd_hist_all}, the concentration of non-stabilizerness in quantum states sampled randomly according to the Haar measure increases with the number of qubits. Consequently, the probability of sampling states with low non-stabilizerness decreases significantly. To address this and obtain a comprehensive description of the critical depolarization probability across all non-stabilizer states, we employed a biased random generator, which was derived by slightly modifying the uniform random generator proposed in \cite{PhysRevA.80.042309}. In that work, the authors propose circuits to generate random quantum states using only rotation gates ($R_Y(\theta)$ and $R_Z(\theta)$) and CNOT gates. Circuits for two-qubit states use 6 independent real parameters (angles for the rotation gates), while circuits for three-qubit states use 14 independent real parameters. Instead of selecting the angles using the distribution provided in that work, we chose them from a uniform distribution. This simple adjustment introduces bias into the random generator, enabling the sampling of quantum states closer to the boundary of the non-stabilizer set, as illustrated in Figs. \ref{fig:CriticalDepolProbvsMagic}\textbf{(c)} and \ref{fig:CriticalDepolProbvsMagic}\textbf{(d)}.  In short, biased sampling is employed as a practical tool to probe rare, highly robust (or non-robust) states that would otherwise be missed.

This approach has a particularly notable impact on the three-qubit case, where the observed NTD range shifts from $[0.186, 0.560]$ to $[0.043, 0.551]$, corresponding to a change in the Critical Depolarization Probability range from $[0.493, 0.733]$ to $[0.241, 0.741]$. As expected, however, the overall effect on the maximum critical depolarization probability remains small due to the increasing concentration of non-stabilizerness at higher qubit counts.

The results for local depolarization in each qubit are shown in Fig.~\ref{fig:LocalCriticalDepolarization}. The threshold depolarization probability \( p \), above which all states lose their non-stabilizerness in the 2- and 3-qubit cases, decreases to approximately \( p \approx 0.49 \) and \( p \approx 0.50 \), respectively.

Interestingly, significant levels of depolarization are required to push generic random states into the stabilizer polytope. This is particularly interesting in the case of local depolarization, which is known to rapidly degrade entanglement and drive states toward separability \cite{aolita2008scaling}. The fact that a high noise threshold is still required to eliminate non-stabilizerness underscores its robustness and supports the interpretation of the NTD as a meaningful resource quantifier.

\begin{figure}[ht!]
    \centering
    \includegraphics[width=\columnwidth]{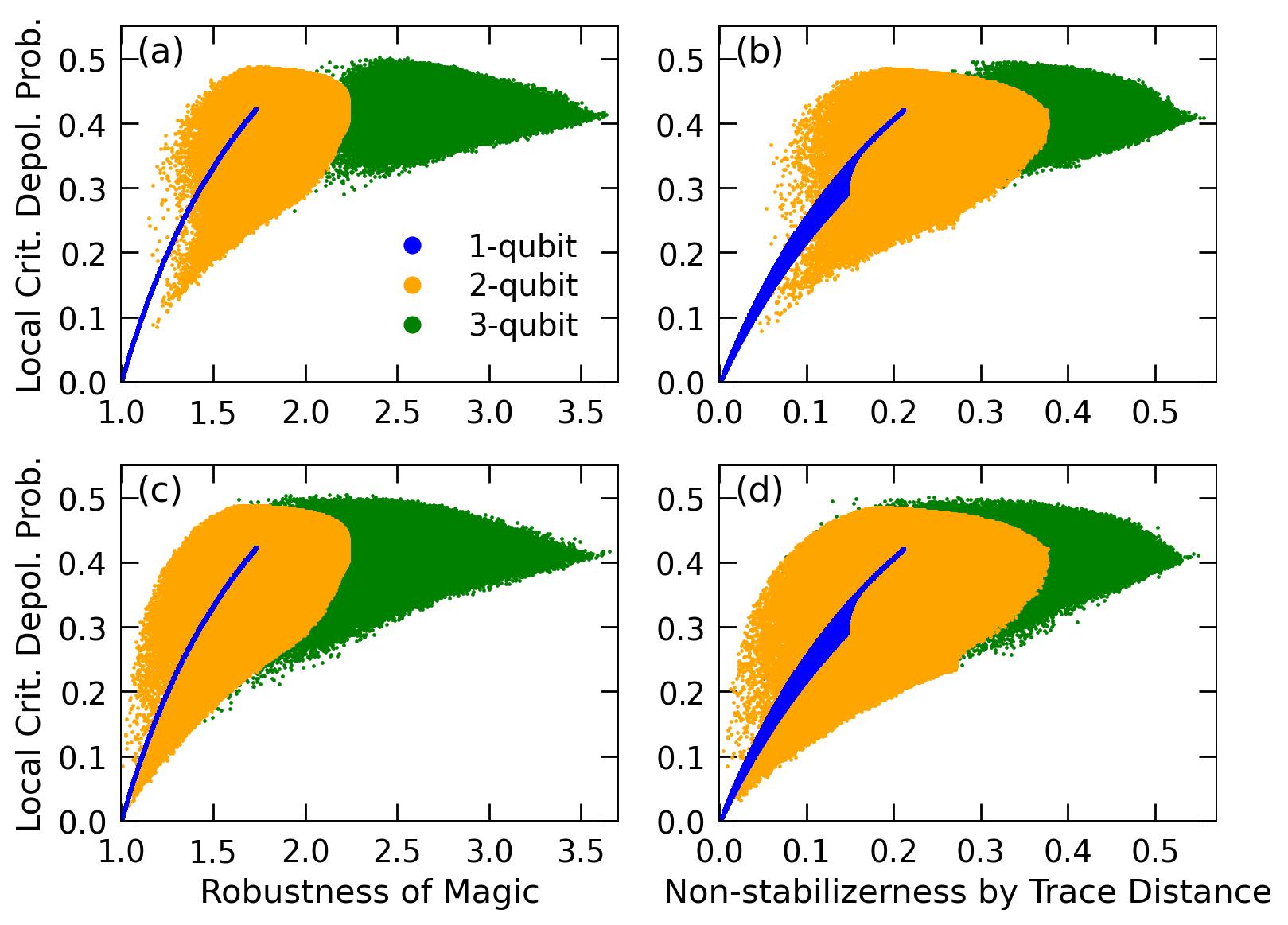}
    \caption{\raggedright
    Local critical depolarization probability for pure states versus: \textbf{(a)} RoM and \textbf{(b)} NTD, using uniformly random quantum states;   \textbf{(c)} RoM and \textbf{(d)} NTD, using biased random quantum states.  
    For each case ($1$, $2$ or $3$ qubits), 1,900,000 states were sampled.  
}

   \label{fig:LocalCriticalDepolarization}
\end{figure}

\section{Non-stabilizerness and entanglement \label{sec:magicent}}

\subsection{Non-stabilizerness versus entanglement in randomly sampled pure states}
To gain further insight into the relationship between entanglement and non-stabilizerness beyond the asymptotic results mentioned above, we have sampled pure quantum states of 2 and 3-qubits using distinct methods. 

For each state, we compute the NTD and its entanglement entropy \( E(\rho) \). For 2-qubit states, this is the von Neumann entropy of the reduced density matrix according to a given bipartition of the state, that is:  
\begin{equation}  
    E(\rho)= S(\rho_A) = -\text{Tr}\left[\rho_A \log(\rho_A)\right]\text{,}  
\end{equation}  
where \(\rho_A\) is the reduced density matrix of a given bipartition \(A\vert B\) of the state \(\rho\). For the 3-qubit case, we employ the mean entropy, which measures the amount of entanglement that the remaining system retains on average after the other two are traced out:

\begin{equation}
    \overline{\mathcal{E}}(\rho_{ABC}) = \dfrac{1}{3}\left[ S(\rho_{A}) + S(\rho_{B})+ S(\rho_{C})\right]\text{ },
\end{equation}
where now \(\rho_A\) is the reduced density matrix of a given bipartition \(A\vert BC\) of the state \(\rho\) and similarly for the other bipartitions.

For the $2$-qubit case, we utilize the same datasets generated for the analysis in Sec. \ref{sec:distribution}, which include states sampled uniformly via the Haar measure and those from the biased generator (see Sec. \ref{subsec:robustness_depolarization}).

For the 3-qubit case, we use three methods: the two above, along with an additional approach that samples states via Hilbert-space walks from stabilizer polytope vertices to the Hoggar state (Eq.~\eqref{eq:hoggar}). Consider the states 
\begin{equation}
\ket{\psi} = \begin{pmatrix}
\psi_0 \\
\vdots \\
\psi_7
\end{pmatrix}
\ket{\varphi} = \begin{pmatrix}
\varphi_0 \\
\vdots \\
\varphi_7
\end{pmatrix}\text{ }.
\label{eq:statevectors}
\end{equation}
Then a walk  from state $\ket{\psi}$ to state $\ket{\varphi}$
is carried out by varying $\epsilon \in [0,1]$ in the expression:

\begin{equation}  
|\xi \rangle = \frac{1}{\sqrt{\sum\limits_{i=0}^{7} \left| \epsilon \psi_i + (1 - \epsilon) \varphi_i \right|^2}} 
\begin{pmatrix}
\epsilon \psi_0 + (1 - \epsilon) \varphi_0 \\
\vdots \\
\epsilon \psi_7 + (1 - \epsilon) \varphi_7 \\
\end{pmatrix}\text{ }.
\label{eq:walk}
\end{equation}
 For the walk, we used a step of $\Delta\epsilon = 0.001$, resulting in 1001 states per walk. Since the number of vertices is 1080, we obtain a total of 1,081,080  states. 

The results for $2$ and $3$-qubits can be seen in Figs.  \ref{fig:entanglementEntropyvsMTD} and \ref{fig:MTDvsMeanEntropy_Pure-3qubit}, respectively. As evident from the plots, certain regions in the trade-off between non-stabilizerness and entanglement are forbidden. Specifically, despite asymptotic results \cite{hayden2006aspects,liu2022many, turkeshi2023pauli} showing that for a large number of qubits, random states exhibit near-maximal non-stabilizerness and entanglement, the few-qubit regime reveals a stark contrast: states with maximal non-stabilizerness cannot also be maximally entangled, and vice versa. This is reflected in the empty upper-right region of the plots.

Another intriguing feature appears in the lower region, where states with intermediate levels of both entanglement and non-stabilizerness seem to be prohibited.  Along the horizontal axis (entanglement entropy), we have special points that are those corresponding to stabilizer states. In the case of $2$-qubits, these correspond to $E=0$ (separable state) and $E=1$ (Bell state), and in the case of $3$-qubits, they correspond to $E=0$ (separable state), $E=2/3$ (biseparable states), and $E=1$ (GHZ state). As we see, near the vertices of the stabilizer polytope, non-stabilizerness is strongly constrained $-$ a result that we will formally establish in the following section using Fannes’ inequality.

\begin{figure}[ht!]
    \centering
    \begin{subfigure}{0.8\columnwidth}
        \centering
        \includegraphics[width=\columnwidth]{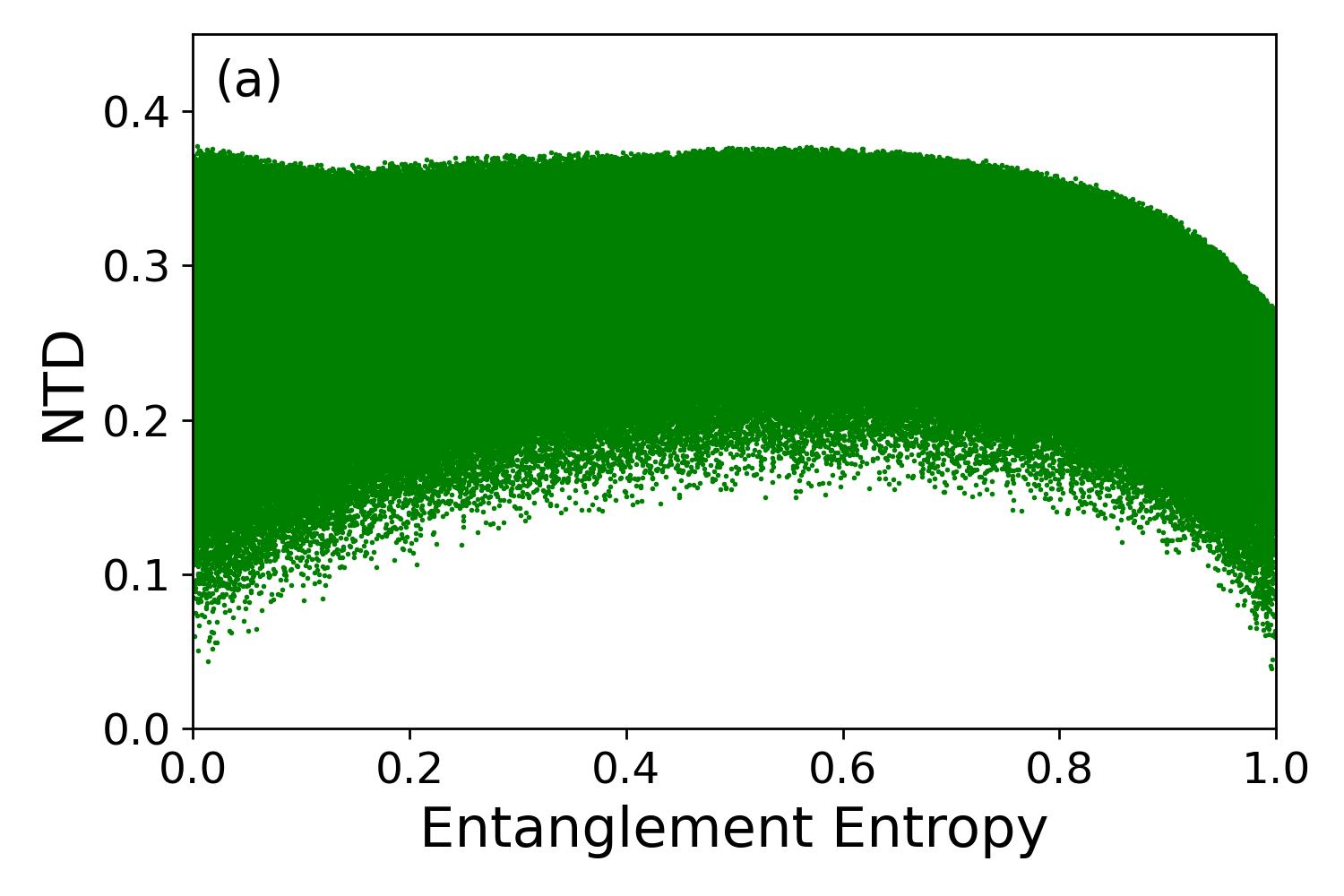}
    \end{subfigure}
   \\ 
   \centering
    \begin{subfigure}{0.8\columnwidth}       
        \includegraphics[width=\columnwidth]{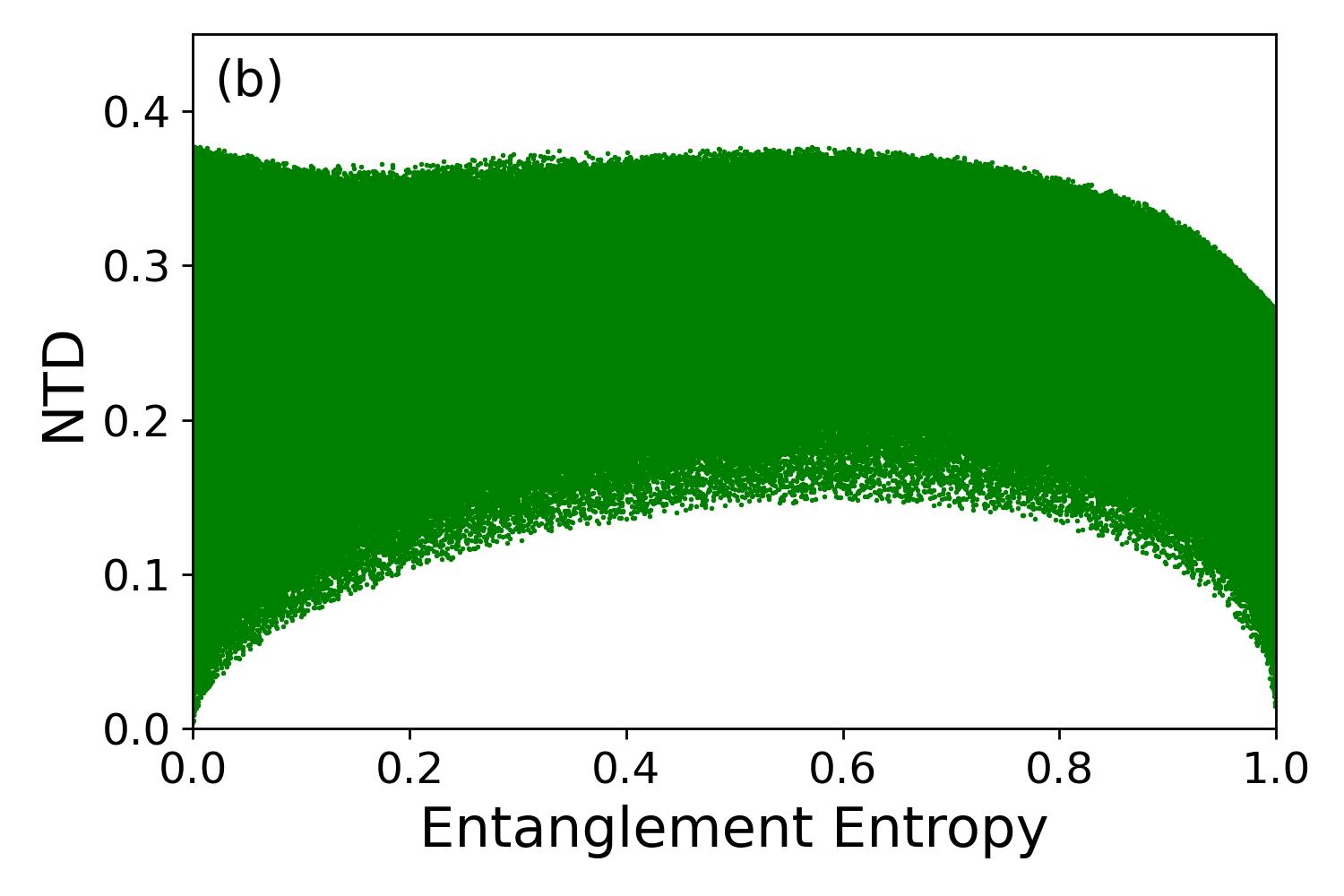}
    \end{subfigure}
\caption{\raggedright Entanglement Entropy vs. Non-stabilizerness by Trace Distance for 2-qubit pure states, using \textbf{(a)} a uniform random generator and \textbf{(b)} a biased random generator (described in Section~\ref{subsec:robustness_depolarization}). For each plot, 2,000,000 states were generated.}
\label{fig:entanglementEntropyvsMTD}
\end{figure}

\begin{figure}[ht!]
    \centering
    \includegraphics[width=\columnwidth]{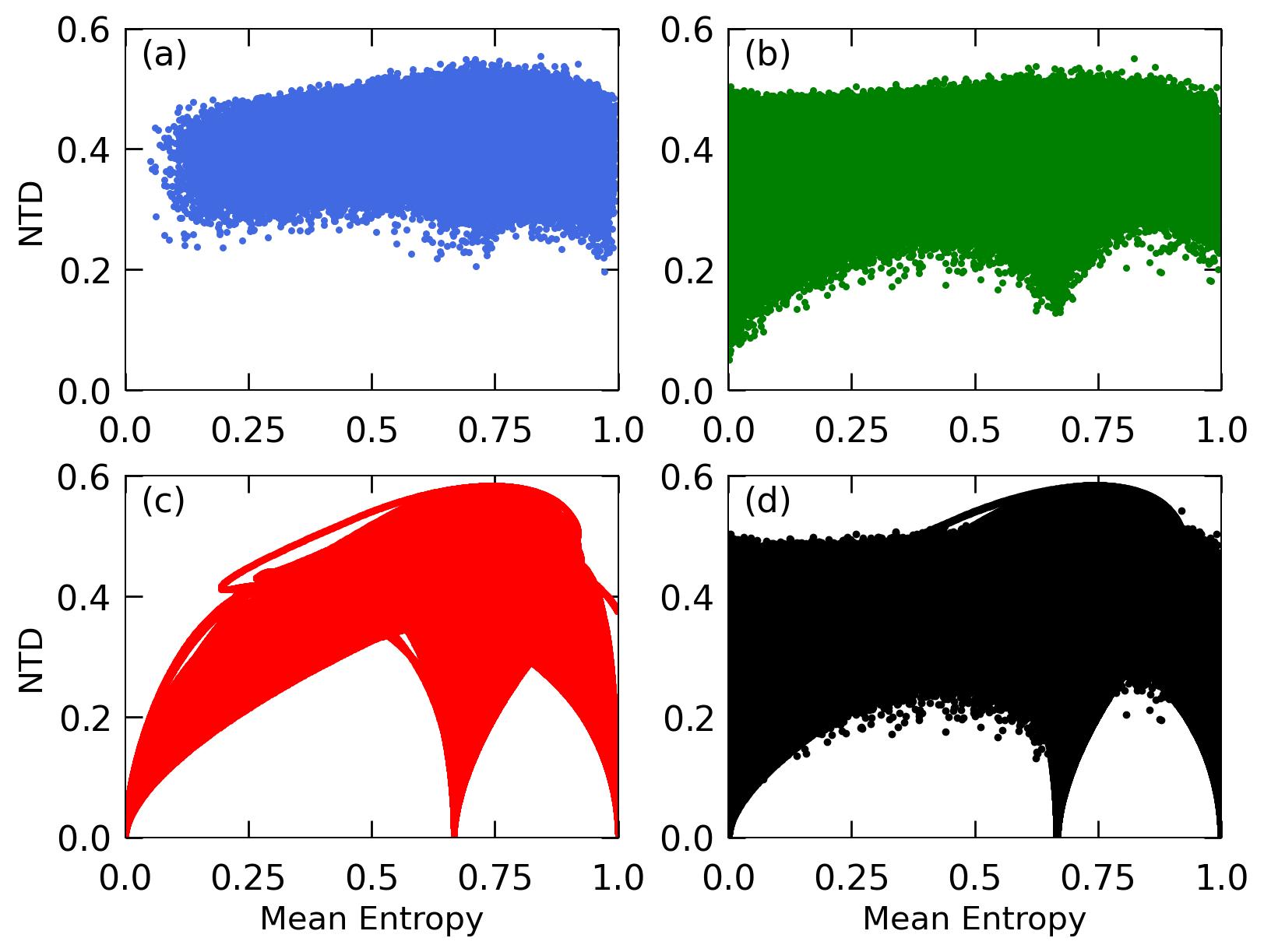}
   \caption{\raggedright NTD vs. Mean Entropy for 3-qubit pure states obtained using 
\textbf{(a)} a uniform random state generator (2,200,000 states), 
\textbf{(b)} a biased random state generator (1,200,000 states), 
\textbf{(c)} walks on the hypersphere from each polytope vertex to a Hoggar state (1,081,080 states), and 
\textbf{(d)} all states combined, showing the most complete relationship between NTD and Mean Entropy.}

    \label{fig:MTDvsMeanEntropy_Pure-3qubit}
\end{figure}

\subsection{Concentration properties of low non-stabilizer states}\label{sec:concentration_properties}

Building on the results discussed above, we will demonstrate that pure states in proximity to the stabilizer polytope inherit its structural properties. To establish this, we first prove a geometric result about polytopes, leveraging the continuity of entropy $-$ first shown by Fannes \cite{fannes1973continuity}, and then enhanced by Audenaert \cite{audenaert2007sharp}:

\begin{lemma}[Fannes-Audenaert inequality]
    Let $\rho, \sigma \in \mathcal D(\mathcal H_{d,n})$ with $ T(\rho, \sigma)=\frac 1 2 \|\rho-\sigma\| = t$. Then:
    \begin{equation}
        |S(\rho)-S(\sigma)| \leq t\log(d^n-1) + H[\{t,1-t\}]\;,\label{eq: Fanneserror}
    \end{equation}
    where
    \begin{equation}
        H[\{t,1-t\}] = -t\log t -(1-t) \log(1-t)\text{ },
    \end{equation}
    is the Shannon entropy of the distribution $\{t,1-t\}$. We will refer to the RHS of Eq.~(\ref{eq: Fanneserror}) as the Fannes error $f(t,d^n)$.
\end{lemma}

\noindent Note that $f(t\to 0, d^n) 
 =0$, and is monotonically increasing for $ t \leq 1-1/d^n$, which can be checked by looking at $\partial_t f(t, d^n) \geq 0$. In this regime, $f(t, d^n) \leq n \log d$.

 Now, consider some polytope of states $P$ whose vertices are pure, as is the case of $P^\mathrm{STAB}_{d,n}$. If a state $\rho$ is known to be close to $P$ in trace distance, then it must also be relatively close to at least one of its vertices, since it must be near a facet. This idea is formalized in the following theorem which is proven in Appendix \ref{apd:concentration}:

\begin{theorem}\label{thm:concentration}
    Let $P$ be a polytope of states built of the convex hull of a set of states:
    \begin{equation}
        P = \mathrm{conv}\{|\psi_1\rangle \langle  \psi_1|, |\psi_2\rangle \langle \psi_2|, \cdots, |\psi_m\rangle \langle \psi_m |\} \subseteq \mathcal D(\mathcal H_{d,n})\text{ }.
    \end{equation}
    Then, there exists  $\varepsilon^* >0$, such that for 
    any pure state $\Psi = \ket{\psi}\bra{\psi}\in \mathcal D(\mathcal H_{d,n})$ with $d_P(\Psi) = \min_{\sigma \in P} T(\rho, \sigma) = \varepsilon \leq \varepsilon^*$, there is a vertex $|\psi_j\rangle \langle \psi_j|$ such that:
    \begin{equation}
        T(\Psi, |\psi_j\rangle \langle \psi_j|) \leq \varepsilon + 2\sqrt{1-h^{-1}[f(\varepsilon,  d^n)]}\;,
    \end{equation}
    with the inverse of the binary entropy function taken in the interval $\mathrm{im} (h^{-1}) = [\frac 1 2 ,1]$.
\label{thm:vertexpoly}
\end{theorem}

\noindent
Hence, for $\varepsilon \to 0$, the distance of $\Psi$ to one of the vertices of the polytope goes to zero. This means that any continuous quantities on states will have their values concentrated to the ones obtained at $P$. Since we are interested in the entanglement entropy, we highlight the following result:
\begin{corollary}
    Let $P$ be a state polytope as written in Theorem \ref{thm:vertexpoly}, and $\Psi = \ket{\psi}\bra{\psi} \in \mathcal D(\mathcal H_{d,n})$ a pure state satisfying the close distance condition, $d_P(\Psi) = \varepsilon \leq \varepsilon^*$. Given a partition $A \subseteq \{1, \cdots,n\}$, it exists a vertex $|\psi_j \rangle \langle \psi_j| $ of $P$ such that, for $\Psi_A = \mathrm{tr}_{\bar A} |\psi\rangle \langle \psi|$ and $\Psi_{j,A} = \mathrm{tr}_{\bar A} |\psi_j\rangle \langle \psi_j|$:
    \begin{equation}
        |S(\Psi_A)-S(\Psi_{j,A})|\leq f(\delta, d^n)\text{ }, \;
    \end{equation}
    where $\delta =\varepsilon + 2\sqrt{1-h^{-1}[f(\varepsilon,  d^n)]}$.
\end{corollary}
\noindent This follows directly from the non-increasing feature of the trace distance under CPTP maps, $T(\rho_A, \sigma_A) \leq T(\rho, \sigma)$, Theorem \ref{thm:vertexpoly} and the Fannes-Audenaert inequality. This result explains the features seen in Figs. \ref{fig:entanglementEntropyvsMTD} and \ref{fig:MTDvsMeanEntropy_Pure-3qubit}, illustrating how sampled states with low non-stabilizerness have close entanglement properties to those of pure stabilizer states. 

\section{Discussion} \label{sec:Discussion}
Understanding the structure of the stabilizer polytope remains a fundamental yet challenging problem due to its rapid growth in complexity with increasing system size.  While asymptotic analyses of random states provide valuable theoretical insights, they often overlook the finer structural details present in intermediate system sizes. These details can be crucial for understanding the constraints on quantum resources in near-term quantum devices.

Motivated by this, we investigated the stabilizer polytope and its geometric properties in small-dimensional quantum systems. Using the trace distance to the stabilizer set as a measure of non-stabilizerness, we explored the structure of the quantum state space and characterized the distribution of non-stabilizerness for both pure and mixed states sampled according to the Haar measure. Additionally, we provided a Bell-like inequality formulation of the stabilizer polytope, extending previous results by finding new useful inequalities and in some cases explicitly connecting the violation of these inequalities to the trace-distance-based quantification of non-stabilizerness. Our findings reveal a strong trade-off between entanglement and non-stabilizerness in few-qubit systems, identifying forbidden regions where high non-stabilizerness strictly limits entanglement and vice versa. Furthermore, leveraging Fannes’ inequality, we derived a general analytical concentration result that formalizes the observed relationship between non-stabilizerness and entanglement for states near the vertices of the stabilizer polytope.

Our results suggest several promising directions for future research. Extending these analyses to higher-dimensional systems, particularly by deriving Bell-like inequalities for arbitrary numbers of qubits, could provide deeper insights into the geometry of the stabilizer set. Additionally, investigating the behavior of non-stabilizerness under realistic noise models may help clarify its resilience as a computational resource. A particularly interesting direction is the potential connection between trace-distance-based non-stabilizerness and the computational complexity of preparing the corresponding quantum states. An important open question is whether the simulability theorems of the RoM can be extended to the NTD \cite{howard2017application}. While our study has focused on few-qubit systems, we hope that our results will inspire further research into these fundamental aspects of stabilizer polytope and its consequences for quantum information processing.

\section*{Acknowledgements}
This work was supported by the Simons Foundation (Grant Number 1023171, RC), the Brazilian National Council for Scientific and Technological Development (CNPq, Grants No.307295/2020-6 and No.403181/2024-0), and the Financiadora de Estudos e Projetos (grant 1699/24 IIF-FINEP). We thank the High-Performance Computing Center (NPAD) at UFRN for providing computational resources. We also thank the Qiskit and PennyLane teams for their SDKs, which were utilized in our simulations.

\appendix

\section{NTD Properties' proofs}\label{apd:mtd_prop}
In this Appendix, we provide the derivations showing that the NTD satisfies the necessary conditions for a valid resource measure, as stated in section \ref{subsubsec:mtd}. The conditions are the following:
\begin{enumerate}
    \item Faithfulness: $\mathcal{M}(\rho) = 0$, if and only if $\rho$ is stabilizer, otherwise, $\mathcal{M}(\rho)>0$,
    \item Monotonicity: for all trace-preserving stabilizer channels $\varepsilon$, $\mathcal{M}(\varepsilon(\rho)) \leq \mathcal{M}(\rho)$,
    \item Convexity: $\mathcal{M}(\sum_i p_i\rho_i) \leq \sum_i p_i\mathcal{M}(\rho_i)$, for $p_i \geq 0$ and $\sum_i p_i = 1$.
    \item Invariance under Clifford Unitaries: $\mathcal{M}(C\rho C^{\dagger}) = \mathcal{M}(\rho)$, for any Clifford unitary $C$.
\end{enumerate}
\begin{proof}
Their demonstrations are the following:
\begin{enumerate}
    \item Notice that NTD is faithful by definition. If the state $\rho$ lies inside the stabilizer polytope, $\rho \in P^{\text{STAB}}_{d,n}$, then there exists $\sigma = \rho \in P^{\text{STAB}}_{d,n}$  such that $\rho -\sigma = 0$, meaning that $\mathcal{M}(\rho) = 0$. If the state $\rho$ is non-stabilizer, then for any $\sigma \in P^{\text{STAB}}_{d,n}$  $\rho -\sigma \geq 0$, implying $\mathcal{M}(\rho)>0$.

    \item The NTD monotonicity follows from the contractiveness of the trace distance under CPTP maps demonstrated in \cite{nielsen2010quantum}. For any trace-preserving stabilizer channel $\varepsilon$,
\begin{equation}
    \mathcal{M}(\varepsilon(\rho)) = \min_{\sigma\in P^{\text{STAB}}_{d,n}} T(\varepsilon(\rho),\sigma).\label{eq: min_stab_channel}
\end{equation}

As $\varepsilon$ is a stabilizer channel, $\varepsilon(\sigma)$ is a stabilizer state for $\sigma \in  P^{\text{STAB}}_{d,n}$. This means that $\{\varepsilon(\sigma)| \sigma \in  P^{\text{STAB}}_{d,n} \}$ is some subset of the stabilizer polytope, which implies the following bound on  Eq.~(\ref{eq: min_stab_channel}): 
\begin{equation}
    \min_{\sigma\in P^{\text{STAB}}_{d,n}}T(\varepsilon(\rho),\sigma) \leq \min_{\sigma\in P^{\text{STAB}}_{d,n}} T(\varepsilon(\rho),\varepsilon(\sigma)). 
\end{equation}

Knowing that these stabilizer channels are CPTP maps, one can use the trace distance's contractiveness under CPTP maps, such that
\begin{equation}
    \min_{\sigma\in P^{\text{STAB}}_{d,n}} T(\varepsilon(\rho),\varepsilon(\sigma)) \leq \min_{\sigma\in P^{\text{STAB}}_{d,n}} T(\rho,\sigma) = \mathcal{M}(\rho)
\end{equation}
proving the monotonicity of the NTD:
\begin{equation}
    \mathcal{M}(\varepsilon(\rho)) \leq \mathcal{M}(\rho).
\end{equation}

\item The NTD convexity follows directly from the trace distance convexity in its first input \cite{nielsen2010quantum}.
\item In turn, the NTD invariance under Clifford unitaries is a direct consequence of the trace distance invariance under unitaries~\cite{nielsen2010quantum}: given a Clifford unitary $C$, we have that
\begin{equation}
    \mathcal{M}(C \rho C^{\dagger}) =  \min_{\sigma\in P^{\text{STAB}}_{d,n}} T(C \rho C^{\dagger} , \sigma).
\end{equation}
Recalling that the stabilizer set is invariant under Clifford unitaries, that is, if $\sigma \in  P^{\text{STAB}}_{d,n}$  then  $C \sigma C^{\dagger}\in  P^{\text{STAB}}_{d,n}$, it follows that
\begin{equation}
    \min_{\sigma\in P^{\text{STAB}}_{d,n}} T(C \rho C^{\dagger} ,\sigma) = \min_{\sigma\in P^{\text{STAB}}_{d,n}} T(C \rho C^{\dagger} , C \sigma C^{\dagger}).
\end{equation}
Using the trace distance's invariance under unitary operations~\cite{nielsen2010quantum}, one can see that
\begin{equation}
    \min_{\sigma\in P^{\text{STAB}}_{d,n}} T(C \rho C^{\dagger} ,C \sigma C^{\dagger}) = \min_{\sigma\in P^{\text{STAB}}_{d,n}} T(\rho  , \sigma) = \mathcal{M}(\rho),
\end{equation}
from which it directly follows that
\begin{equation}
    \mathcal{M}(C \rho C^{\dagger}) = \mathcal{M}(\rho).
\end{equation}
\end{enumerate}
This completes the proof showing that the trace distance is a valid measure of non-stabilizerness. 
\end{proof}

\section{Concentration properties of low non-stabilizer states' proofs}\label{apd:concentration}
In this Appendix, we provide the proof for Theorem \ref{thm:concentration} from section \ref{sec:concentration_properties}. To do so, we first establish a preliminary lemma relating the Von Neumann entropy of a state to its maximal eigenvalue.

 \begin{lemma}\label{lem:aux_entropy}
     Let $\rho$ be any state with dimension $d \geq 2$, and $\{\lambda_i\}$ its spectrum. Then,  (i) $S(\rho) \geq h(\lambda_\mathrm{max})$, where  $h(\lambda_\mathrm{max}) = H(\{\lambda_\mathrm{max}, 1-\lambda_\mathrm{max}\}) = -\lambda_\mathrm{max} \log \lambda_\mathrm{max} - (1-\lambda_\mathrm{max})\log(1-\lambda_\mathrm{max})$ is the binary entropy function of the maximum eigenvalue $\lambda_\mathrm{max}\equiv \max_i \{\lambda_i\}$ of $\rho$ and (ii)
      if its Von Neumann entropy satisfies $S(\rho) < 1 $, then  $\lambda_\mathrm{max} > 1/2$.
      \label{lem:vonneumanntoshannon}
 \end{lemma}
 \begin{proof}
     Denote $\{\lambda_i\}_i$ as the spectrum of $\rho$. We know then:
     \begin{equation}
         S(\rho) = H(\{\lambda_i\}_i) = -\sum_i \lambda_i \log \lambda_i\;.
     \end{equation}
    We will show that both results follow from standard information-theoretic properties of the Shannon entropy.
     
     (i) Note that we can write:
     \begin{equation}
         H(\{\lambda_i\}_i) -h(\lambda_\mathrm{max}) = g\;,
     \end{equation}
     where:
     \begin{align}
         g &= (1-\lambda_\mathrm{max}) \log(1-\lambda_\mathrm{max}) - \sum_{i \neq \lambda_\mathrm{max}}\lambda_i \log \lambda_i\\
         & =  - \sum_{i \neq \lambda_\mathrm{max}}\lambda_i \log  \left(\frac{\lambda_i}{1-\lambda_\mathrm{max}}\right)\\
         &= (1-\lambda_\mathrm{max}) H\left(\left \{\frac{\lambda_i}{1-\lambda_\mathrm{max}} \right \}\right) \geq 0\;,
     \end{align}
     where we used that $\lambda_\mathrm{max} + \sum_{i \neq \mathrm{max}} \lambda_i=1$ on the first to the second line. Therefore $S(\rho) \geq h(\lambda_\mathrm{max})$ .
     
     (ii) follows from the Schur-concavity of the Shannon entropy. That is, given two probability distributions $ \{\lambda_i\}$ and $ \{\lambda^\prime_i\}$ such that $\{ \lambda_i\} \succ \{\lambda^\prime_i\}$, meaning that:
     \begin{equation}
         \sum_{i=1}^{k} \lambda_i^{\downarrow} \geq \sum_{i=1}^{k} (\lambda_i^\prime)^{\downarrow }\quad  ; \quad \forall k \in \{1,2,\cdots, D-1\}\;,
        \label{eq:majorization}
     \end{equation}
     where $[\lambda_i^{\downarrow }]_i$ and $[(\lambda^\prime_i)^{\downarrow}]_i$ denote both sets written in decreasing order, then it is known that the Shannon entropy satisfies $H(\{\lambda_i\})  \leq H(\{\lambda^\prime_i\} )$. 

     First, we will check that the probability distribution $ \{Q_1, Q_2, Q_3, \cdots , Q_d\} = (1/2, 1/2,0, \cdots, 0 )$ majorizes any $\{\lambda_i\}_i$ such that $\max_i \{\lambda_i\} \leq 1/2$. In this case, we know that $\sum_{i \neq \mathrm{max}}\lambda_i \geq 1/2$. The first condition on Eq. (\ref{eq:majorization}), for $k=1$, is given as:
     \begin{equation}
          Q^{\downarrow}_1 \geq \lambda_1^{\downarrow}\Rightarrow \lambda_\mathrm{max} \leq \frac 1 2\;,
     \end{equation}
     which is indeed true. All the other conditions are trivially satisfied, since $\sum_{i=1}^{k > 1}Q^\downarrow_i =1$. Therefore, the corresponding Shannon entropy satisfies:
     \begin{equation}
         H(\{\lambda_i\}) \geq H(\{Q_i\}) = h(1/2) =1\;.
     \end{equation}
    Hence, if $H(\{\lambda_i\}) < 1$, we must have $\lambda_\mathrm{max} > 1/2$. 
 \end{proof}

Equipped with Lemma \ref{lem:aux_entropy}, we now proceed to the proof of the main concentration result:

 \begin{proof}
    Let $\sigma_{\Psi} = \mathrm{argmin}_{\sigma \in P}T(\Psi, \sigma)$ be the nearest state in $P$ to $\Psi$, unique due to convexity. Since $P$ is a polytope, $\sigma_{\Psi}$ lies in a facet $F$, so we can write
    \begin{align}
    \sigma_{\Psi} = \sum_{i=1}^{|F|} p_i |\psi_i\rangle \langle\psi_i|\;.  
    \end{align}
    where $|F|$ is the number of vertices spanning $F$. By the Fannes-Audenaert inequality, its entropy must be similar  w.r.t the state $\Psi =  \ket{\psi} \bra{\psi}$, which vanishes:
    \begin{align}
    S(\sigma_\Psi) \leq  f(\varepsilon,d^n)\;. 
    \end{align}
    By the results of Lemma \ref{lem:vonneumanntoshannon}, this implies that we have:
    \begin{equation}
       h( \lambda_\mathrm{max} ) \leq f(\varepsilon, d^n)\;,
    \end{equation}
    with $\lambda_\mathrm{max}$  the maximal eigenvalue of $\sigma_\Psi$. We thus define $\varepsilon_* \leq  1-1/d^n$ as the value such that $f(\varepsilon, d^n)  < 1$ for all $\varepsilon \leq \varepsilon^*$. This can be attained since we have that $f(\varepsilon \leq \varepsilon_*, d^n) \leq n \log d$.
    
    Hence, for $\varepsilon \leq \varepsilon^*$, $S(\sigma_\Psi) < 1$  . Lemma \ref{lem:vonneumanntoshannon} implies $\lambda_\mathrm{max} > 1/2$, on which the binary entropy function is one-to-one and monotonically decreasing. Therefore:
    \begin{equation}
        \lambda_\mathrm{max} \geq h^{-1}[f(\varepsilon, d^n)]\;.
        \label{eq:maxeigenvaluebound}
    \end{equation}
    
    Now, let $\ket{\psi_\mathrm{max}}$ be an eigenvector of $\sigma_\Psi$ with eigenvalue $\lambda_\mathrm{max}$. Then, by the Fuchs- van de Graaf inequality \cite{nielsen2010quantum}:
    \begin{equation}
        T(\sigma_\Psi, \ket{\psi_\mathrm{max}} \bra{\psi_\mathrm{max}})  \leq \sqrt{1- \bra{\psi_\mathrm{max}} \sigma_\Psi \ket{\psi_\mathrm{max}}}\;.
    \end{equation}
    Since $\bra{\psi_\mathrm{max}} \sigma_\Psi \ket{\psi_\mathrm{max}} = \lambda_\mathrm{max}$ and $f(x) = \sqrt{1-x}$ is monotonically decreasing for $x \in [0,1]$, we then derive:
    \begin{equation}
    T(\sigma_\Psi,\ket{\psi_\mathrm{max}} \bra{\psi_\mathrm{max}}) \leq \sqrt{1- h^{-1}[f(\varepsilon, d^n)]}\;,
    \end{equation}
    which implies that $\sigma_\Psi$ becomes pure as $\varepsilon \to 0$. We now proceed to show that the distance of $\ket{\psi_\mathrm{max}}$ and the nearest vertex $\ket{\psi_{j_\mathrm{max}}}$ is also upper bounded by the same quantity.
    We have:
    \begin{equation}
        T(\ket{\psi_\mathrm{max}} \bra{\psi_\mathrm{max}}, \ket{\psi_j} \bra{\psi_j} ) = \sqrt{1-|\langle \psi_\mathrm{max}| \psi_j\rangle|^2}\;,
    \end{equation}
    note that:
    \begin{equation}
        \lambda_\mathrm{max} = \sum_j p_j |\langle \psi_\mathrm{max}| \psi_j\rangle|^2 \leq |\langle \psi_\mathrm{max}| \psi_{j_\mathrm{max}}\rangle|^2\;,
    \end{equation}
    where $\psi_{j_\mathrm{max}} = \mathrm{argmax}_j |\langle \psi_\mathrm{max}|\psi_j\rangle|^2$ is the state with maximum fidelity with $\psi_\mathrm{max}$ in the facet. By Eq. (\ref{eq:maxeigenvaluebound}), it follows that:
    \begin{equation}
        T(\ket{\psi_\mathrm{max}} \bra{\psi_\mathrm{max}}, \ket{\psi_{j_\mathrm{max}}} \bra{\psi_{j_\mathrm{max}}} ) \leq \sqrt{1-h^{-1}[f(\varepsilon, d^n)]}\;.
    \end{equation}
    Now, we can use the triangular inequality repeatedly on the sequence of states $\Psi \rightarrow \sigma_\Psi \rightarrow \psi_\mathrm{max} \to \psi_{j_\mathrm{max}}$, and we show the desired bound by summing the three distances:
    \begin{align}
        T(\Psi, \ket{\psi_{j_\mathrm{max}}}\bra{\psi_{j_\mathrm{max}}}) \leq \varepsilon + 2\sqrt{1-h^{-1}[f(\varepsilon,  d^n)]}\;.
    \end{align}
\end{proof}

\bibliography{refs}

@misc{ipek2023,
      title={The degenerate vertices of the $2$-qubit $\Lambda$-polytope and their update rules}, 
      author={Selman Ipek and Cihan Okay},
      year={2023},
      eprint={2312.10734},
      archivePrefix={arXiv},
      primaryClass={quant-ph},
      url={https://arxiv.org/abs/2312.10734}, 
}

@mastersthesis{Heimendahl2019,
  author  = {Arne Heimendahl},
  title   = {The Stabilizer Polytope and Contextuality for Qubit Systems},
  school  = {University of Cologne},
  year    = {2019},
  month   = {6},
  date    = {2019-06-04},
  note    = {Supervised by Frank Vallentin and David Gross}
}

@article{Delfosse_2017,
doi = {10.1088/1367-2630/aa8fe3},
url = {https://dx.doi.org/10.1088/1367-2630/aa8fe3},
year = {2017},
month = {dec},
publisher = {IOP Publishing},
volume = {19},
number = {12},
pages = {123024},
author = {Delfosse, Nicolas and Okay, Cihan and Bermejo-Vega, Juan and Browne, Dan E and Raussendorf, Robert},
title = {Equivalence between contextuality and negativity of the Wigner function for qudits},
journal = {New Journal of Physics}
}

@article{Zurel2024hiddenvariablemodel,
  doi = {10.22331/q-2024-04-30-1323},
  url = {https://doi.org/10.22331/q-2024-04-30-1323},
  title = {Hidden variable model for quantum computation with magic states on qudits of any dimension},
  author = {Zurel, Michael and Okay, Cihan and Raussendorf, Robert and Heimendahl, Arne},
  journal = {{Quantum}},
  issn = {2521-327X},
  publisher = {{Verein zur F{\"{o}}rderung des Open Access Publizierens in den Quantenwissenschaften}},
  volume = {8},
  pages = {1323},
  month = apr,
  year = {2024}
}

@article{10.5555/2012098.2012105,
  title={Quantum universality by state distillation},
  author={Ben W Reichardt},
  journal={Quantum Inf. Comput.},
  year={2006},
  volume={9},
  pages={1030-1052},
  url={https://api.semanticscholar.org/CorpusID:16854370}
}

@article{Hostens_2005,
   title={Stabilizer states and Clifford operations for systems of arbitrary dimensions and modular arithmetic},
   author={Hostens, Erik and Dehaene, Jeroen and De Moor, Bart},
   journal={Phys. Rev. A},
   volume={71},
   issue={4},
   pages={042315},
   numpages={9},
   year={2005},
   month={Apr},
   publisher={American Physical Society},
   doi={10.1103/PhysRevA.71.042315},
   url={https://link.aps.org/doi/10.1103/PhysRevA.71.042315}
}

@article{van2005local,
  title = {Local unitary versus local Clifford equivalence of stabilizer states},
  author = {Van den Nest, Maarten and Dehaene, Jeroen and De Moor, Bart},
  journal = {Phys. Rev. A},
  volume = {71},
  issue = {6},
  pages = {062323},
  numpages = {7},
  year = {2005},
  month = {Jun},
  publisher = {American Physical Society},
  doi = {10.1103/PhysRevA.71.062323},
  url = {https://link.aps.org/doi/10.1103/PhysRevA.71.062323}
}

@misc{avis2001lrs,
   author={David Avis},
   title={{lrs Homepage}},
   year={2001},
   note={{School of Computer Science}, McGill University, Canada},
   howpublished={\url{http://cgm.cs.mcgill.ca/~avis/C/lrs.html}}
}

@incollection{avis2002canonical,
   title={On canonical representations of convex polyhedra},
   author={Avis, David and Fukuda, Komei and Picozzi, Stefano},
   booktitle={Mathematical software},
   pages={350--360},
   year={2002},
   publisher={World Scientific},
   doi={10.1142/9789812777171_0037},
   url={https://www.worldscientific.com/doi/abs/10.1142/9789812777171_0037},
}

@book{skrzypczyk2023semidefinite,
   author={Skrzypczyk, Paul and Cavalcanti, Daniel},
   title={Semidefinite Programming in Quantum Information Science},
   publisher={IOP Publishing},
   year={2023},
   series={2053-2563},
   isbn={978-0-7503-3343-6},
   url={https://dx.doi.org/10.1088/978-0-7503-3343-6},
   doi={10.1088/978-0-7503-3343-6}
}

@article{PhysRevA.98.032304,
  title = {Normal form for single-qutrit Clifford {$+T$} operators and synthesis of single-qutrit gates},
  author = {Prakash, Shiroman and Jain, Akalank and Kapur, Bhakti and Seth, Shubangi},
  journal = {Phys. Rev. A},
  volume = {98},
  issue = {3},
  pages = {032304},
  numpages = {9},
  year = {2018},
  month = {Sep},
  publisher = {American Physical Society},
  doi = {10.1103/PhysRevA.98.032304},
  url = {https://link.aps.org/doi/10.1103/PhysRevA.98.032304}
}

@Article{Howard2014,
author={Howard, Mark
and Wallman, Joel
and Veitch, Victor
and Emerson, Joseph},
title={Contextuality supplies the `magic' for quantum computation},
journal={Nature},
year={2014},
month={Jun},
day={01},
volume={510},
number={7505},
pages={351-355},
issn={1476-4687},
doi={10.1038/nature13460},
url={https://doi.org/10.1038/nature13460}
}

@misc{gottesman1998heisenberg,
      title={The {H}eisenberg Representation of Quantum Computers}, 
      author={Daniel Gottesman},
      year={1998},
      eprint={quant-ph/9807006},
      archivePrefix={arXiv},
      primaryClass={quant-ph},
      url={https://arxiv.org/abs/quant-ph/9807006}, 
}

@article{haug2023quantifying,
  title = {Quantifying nonstabilizerness of matrix product states},
  author = {Haug, Tobias and Piroli, Lorenzo},
  journal = {Phys. Rev. B},
  volume = {107},
  issue = {3},
  pages = {035148},
  numpages = {10},
  year = {2023},
  month = {Jan},
  publisher = {American Physical Society},
  doi = {10.1103/PhysRevB.107.035148},
  url = {https://link.aps.org/doi/10.1103/PhysRevB.107.035148}
}

@article{fannes1973continuity,
   author={Fannes, M.},
   title={A continuity property of the entropy density for spin lattice systems},
   journal={Communications in Mathematical Physics},
   year={1973},
   volume={31},
   number={4},
   pages={291-294},
   sn={1432-0916},
   doi={10.1007/BF01646490},
   url={https://doi.org/10.1007/BF01646490}
}

@article{zyczkowski2001induced,
   doi={10.1088/0305-4470/34/35/335},
   url={https://dx.doi.org/10.1088/0305-4470/34/35/335},
   year={2001},
   month={aug},
   volume={34},
   number={35},
   pages={7111},
   author={Karol Zyczkowski and Hans-Jürgen Sommers},
   title={Induced measures in the space of mixed quantum states},
   journal={Journal of Physics A: Mathematical and General}
}

@article{brito2018quantifying,
   title={Quantifying Bell nonlocality with the trace distance},
   author={Brito, S. G. A. and Amaral, B. and Chaves, R.},
   journal={Phys. Rev. A},
   volume={97},
   issue={2},
   pages={022111},
   numpages={11},
   year={2018},
   month={Feb},
   publisher={American Physical Society},
   doi={10.1103/PhysRevA.97.022111},
   url={https://link.aps.org/doi/10.1103/PhysRevA.97.022111}
}

@article{aolita2008scaling,
  title={Scaling laws for the decay of multiqubit entanglement},
  author={Aolita, Leandro and Chaves, Rafael and Cavalcanti, Daniel and Ac{\'\i}n, A and Davidovich, Luiz},
  journal={Physical Review Letters},
  volume={100},
  number={8},
  pages={080501},
  year={2008},
  publisher={APS},
  url = {https://link.aps.org/doi/10.1103/PhysRevLett.100.080501}
}

@misc{dowling2025bridging,
   title={Bridging Entanglement and Magic Resources through Operator Space},
   author={Neil Dowling and Kavan Modi and Gregory A. L. White},
   year={2025},
   eprint={2501.18679},
   archivePrefix={arXiv},
   primaryClass={quant-ph},
   url={https://arxiv.org/abs/2501.18679}
}

@article{howard2015maximum,
  title = {Maximum nonlocality and minimum uncertainty using magic states},
  author = {Howard, Mark},
  journal = {Phys. Rev. A},
  volume = {91},
  issue = {4},
  pages = {042103},
  numpages = {10},
  year = {2015},
  month = {Apr},
  publisher = {American Physical Society},
  doi = {10.1103/PhysRevA.91.042103},
  url = {https://link.aps.org/doi/10.1103/PhysRevA.91.042103}
}

@misc{meyer2024bell,
   title={Qudit Clauser-Horne-Shimony-Holt Inequality and Nonlocality from Wigner Negativity},
   author={Uta Isabella Meyer and Ivan Šupić and Damian Markham and Frédéric Grosshans},
   year={2025},
   eprint={2405.14367},
   archivePrefix={arXiv},
   primaryClass={quant-ph},
   url={https://arxiv.org/abs/2405.14367}
}

@InProceedings{claudet2025local,
  author =	{Claudet, Nathan and Perdrix, Simon},
  title =	{{Local Equivalence of Stabilizer States: A Graphical Characterisation}},
  booktitle =	{42nd International Symposium on Theoretical Aspects of Computer Science (STACS 2025)},
  pages =	{27:1--27:18},
  series =	{Leibniz International Proceedings in Informatics (LIPIcs)},
  ISBN =	{978-3-95977-365-2},
  ISSN =	{1868-8969},
  year =	{2025},
  volume =	{327},
  editor =	{Beyersdorff, Olaf and Pilipczuk, Micha{\l} and Pimentel, Elaine and Thang, Nguyen Kim},
  publisher =	{Schloss Dagstuhl -- Leibniz-Zentrum f{\"u}r Informatik},
  address =	{Dagstuhl, Germany},
  URL =		{https://drops.dagstuhl.de/entities/document/10.4230/LIPIcs.STACS.2025.27},
  URN =		{urn:nbn:de:0030-drops-228527},
  doi =		{10.4230/LIPIcs.STACS.2025.27}
}

@article{mermin1990extreme,
  title={Extreme quantum entanglement in a superposition of macroscopically distinct states},
  author={Mermin, N David},
  journal={Physical Review Letters},
  volume={65},
  number={15},
  pages={1838},
  year={1990},
  publisher={APS},
  doi = {10.1103/PhysRevLett.65.1838},
  url = {https://link.aps.org/doi/10.1103/PhysRevLett.65.1838}
}

@misc{cusomano2025,
      title={Non-stabilizerness and violations of CHSH inequalities}, 
      author={Stefano Cusumano and Lorenzo Campos Venuti and Simone Cepollaro and Gianluca Esposito and Daniele Iannotti and Barbara Jasser and Jovan Odavi\' c and Michele Viscardi and Alioscia Hamma},
      year={2025},
      eprint={2504.03351},
      archivePrefix={arXiv},
      primaryClass={quant-ph},
      url={https://arxiv.org/abs/2504.03351}, 
}

@article{howard2012nonlocality,
   title={Nonlocality as a benchmark for universal quantum computation in Ising anyon topological quantum computers},
   author={Howard, Mark and Vala, Jiri},
   journal={Phys. Rev. A},
   volume={85},
   issue={2},
   pages={022304},
   numpages={7},
   year={2012},
   month={Feb},
   publisher={American Physical Society},
   doi={10.1103/PhysRevA.85.022304},
   url={https://link.aps.org/doi/10.1103/PhysRevA.85.022304}
}

@article{clauser1969proposed,
   title={Proposed Experiment to Test Local Hidden-Variable Theories},
   author={Clauser, John F. and Horne, Michael A. and Shimony, Abner and Holt, Richard A.},
   journal={Phys. Rev. Lett.},
   volume={23},
   issue={15},
   pages={880--884},
   year={1969},
   month={Oct},
   publisher={American Physical Society},
   doi={10.1103/PhysRevLett.23.880},
   url={https://link.aps.org/doi/10.1103/PhysRevLett.23.880}
}

@article{bravyi2016improved,
  title = {Improved Classical Simulation of Quantum Circuits Dominated by Clifford Gates},
  author = {Bravyi, Sergey and Gosset, David},
  journal = {Phys. Rev. Lett.},
  volume = {116},
  issue = {25},
  pages = {250501},
  numpages = {5},
  year = {2016},
  month = {Jun},
  publisher = {American Physical Society},
  doi = {10.1103/PhysRevLett.116.250501},
  url = {https://link.aps.org/doi/10.1103/PhysRevLett.116.250501}
}

@article{cao2024gravitational,
  title={Gravitational back-reaction is magical},
  author={Cao, CJ and Cheng, G and Hamma, A and Leone, L and Munizzi, W and Oliviero, SFE},
  journal={arXiv preprint arXiv:2403.07056},
  url={https://doi.org/10.48550/arXiv.2403.07056},
  year={2024}
}

@misc{cepollaro2024harvesting,
      title={Harvesting stabilizer entropy and non-locality from a quantum field}, 
      author={S. Cepollaro and S. Cusumano and A. Hamma and G. Lo Giudice and J. Odavic},
      year={2025},
      eprint={2412.11918},
      archivePrefix={arXiv},
      primaryClass={quant-ph},
      url={https://arxiv.org/abs/2412.11918}, 
}

@article{haug2023scalable,
  title={Scalable measures of magic resource for quantum computers},
  author={Haug, Tobias and Kim, MS},
  journal={PRX Quantum},
  volume={4},
  number={1},
  pages={010301},
  year={2023},
  url={https://doi.org/10.1103/PRXQuantum.4.010301},
  publisher={APS}
}

@misc{wagner2024certifyingnonstabilizernessquantumprocessors,
      title={Certifying nonstabilizerness in quantum processors}, 
      author={Rafael Wagner and Filipa C. R. Peres and Emmanuel Zambrini Cruzeiro and Ernesto F. Galvão},
      year={2024},
      eprint={2404.16107},
      archivePrefix={arXiv},
      primaryClass={quant-ph},
      url = {https://arxiv.org/abs/2404.16107}
}

@article{obst2024wigner,
    author = {Obst, Valentin and Heimendahl, Arne and Singal, Tanmay and Gross, David},
    title = {Wigner’s theorem for stabilizer states and quantum designs},
    journal = {Journal of Mathematical Physics},
    volume = {65},
    number = {11},
    pages = {112202},
    year = {2024},
    month = {11},
    issn = {0022-2488},
    doi = {10.1063/5.0222546},
    url = {https://doi.org/10.1063/5.0222546}
    
}

@article{vairogs2025extracting,
  title={Extracting randomness from magic quantum states},
  author={Vairogs, Christopher and Yan, Bin},
  journal={Physical Review Research},
  volume={7},
  number={2},
  pages={L022069},
  year={2025},
  publisher={APS},
  url={https://doi.org/10.1103/3ttm-vhdt}
}

@article{zamora2025prepare,
  title={Prepare-and-Magic: Semi-Device Independent Magic Certification in the Prepare-and-Measure Scenario},
  author={Zamora, Santiago and Macedo, Rafael A and Sarubi, Tailan S and Alves, Mois{\'e}s and Poderini, Davide and Chaves, Rafael},
  journal={arXiv preprint arXiv:2506.02226},
  year={2025},
  url={https://arxiv.org/abs/2506.02226}
}

@misc{macedo2025witnessing,
      title={Witnessing Magic with Bell inequalities}, 
      author={Rafael A. Macedo and Patrick Andriolo and Santiago Zamora and Davide Poderini and Rafael Chaves},
      year={2025},
      eprint={2503.18734},
      archivePrefix={arXiv},
      primaryClass={quant-ph},
      url={https://arxiv.org/abs/2503.18734}, 
}

@article{heinrich2019robustness,
  doi = {10.22331/q-2019-04-08-132},
  url = {https://doi.org/10.22331/q-2019-04-08-132},
  title = {Robustness of {M}agic and {S}ymmetries of the {S}tabiliser {P}olytope},
  author = {Heinrich, Markus and Gross, David},
  journal = {{Quantum}},
  issn = {2521-327X},
  publisher = {{Verein zur F{\"{o}}rderung des Open Access Publizierens in den Quantenwissenschaften}},
  volume = {3},
  pages = {132},
  month = apr,
  year = {2019}
}

@article{pashayan2015estimating,
  title = {Estimating Outcome Probabilities of Quantum Circuits Using Quasiprobabilities},
  author = {Pashayan, Hakop and Wallman, Joel J. and Bartlett, Stephen D.},
  journal = {Phys. Rev. Lett.},
  volume = {115},
  issue = {7},
  pages = {070501},
  numpages = {5},
  year = {2015},
  month = {Aug},
  publisher = {American Physical Society},
  doi = {10.1103/PhysRevLett.115.070501},
  url = {https://link.aps.org/doi/10.1103/PhysRevLett.115.070501}
}

@article{hall1998random,
   title={Random quantum correlations and density operator distributions},
   journal={Physics Letters A},
   volume={242},
   number={3},
   pages={123-129},
   year={1998},
   issn={0375-9601},
   doi={https://doi.org/10.1016/S0375-9601(98)00190-X},
   url={https://www.sciencedirect.com/science/article/pii/S037596019800190X},
   author={Michael J.W. Hall}
}

@article{hayden2006aspects,
   author={Hayden, Patrick and Leung, Debbie W. and Winter, Andreas},
   title={Aspects of Generic Entanglement},
   journal={Communications in Mathematical Physics},
   year={2006},
   volume={265},
   number={1},
   pages={95-117},
   sn={1432-0916},
   doi={10.1007/s00220-006-1535-6},
   url={https://doi.org/10.1007/s00220-006-1535-6}
}

@article{liu2022many,
  title = {Many-Body Quantum Magic},
  author = {Liu, Zi-Wen and Winter, Andreas},
  journal = {PRX Quantum},
  volume = {3},
  issue = {2},
  pages = {020333},
  numpages = {18},
  year = {2022},
  month = {May},
  publisher = {American Physical Society},
  doi = {10.1103/PRXQuantum.3.020333},
  url = {https://link.aps.org/doi/10.1103/PRXQuantum.3.020333}
}

@article{haug2024efficient,
  title = {Efficient Quantum Algorithms for Stabilizer Entropies},
  author = {Haug, Tobias and Lee, Soovin and Kim, M. S.},
  journal = {Phys. Rev. Lett.},
  volume = {132},
  issue = {24},
  pages = {240602},
  numpages = {7},
  year = {2024},
  month = {Jun},
  publisher = {American Physical Society},
  doi = {10.1103/PhysRevLett.132.240602},
  url = {https://link.aps.org/doi/10.1103/PhysRevLett.132.240602}
}

@article{hein2004multiparty,
  title = {Multiparty entanglement in graph states},
  author = {Hein, M. and Eisert, J. and Briegel, H. J.},
  journal = {Phys. Rev. A},
  volume = {69},
  issue = {6},
  pages = {062311},
  numpages = {20},
  year = {2004},
  month = {Jun},
  publisher = {American Physical Society},
  doi = {10.1103/PhysRevA.69.062311},
  url = {https://link.aps.org/doi/10.1103/PhysRevA.69.062311}
}

@article{Veitch_2014,
doi = {10.1088/1367-2630/16/1/013009},
url = {https://dx.doi.org/10.1088/1367-2630/16/1/013009},
year = {2014},
month = {jan},
publisher = {IOP Publishing},
volume = {16},
number = {1},
pages = {013009},
author = {Victor Veitch and S A Hamed Mousavian and Daniel Gottesman and Joseph Emerson},
title = {The resource theory of stabilizer quantum computation},
journal = {New Journal of Physics}
}

@article{PhysRevA.87.030301,
  title = {Process verification of two-qubit quantum gates by randomized benchmarking},
  author = {C\'orcoles, A. D. and Gambetta, Jay M. and Chow, Jerry M. and Smolin, John A. and Ware, Matthew and Strand, Joel and Plourde, B. L. T. and Steffen, M.},
  journal = {Phys. Rev. A},
  volume = {87},
  issue = {3},
  pages = {030301},
  numpages = {4},
  year = {2013},
  month = {Mar},
  publisher = {American Physical Society},
  doi = {10.1103/PhysRevA.87.030301},
  url = {https://link.aps.org/doi/10.1103/PhysRevA.87.030301}
}

@article{Heinrich2019robustnessofmagic,
  doi = {10.22331/q-2019-04-08-132},
  url = {https://doi.org/10.22331/q-2019-04-08-132},
  title = {Robustness of {M}agic and {S}ymmetries of the {S}tabiliser {P}olytope},
  author = {Heinrich, Markus and Gross, David},
  journal = {{Quantum}},
  issn = {2521-327X},
  publisher = {{Verein zur F{\"{o}}rderung des Open Access Publizierens in den Quantenwissenschaften}},
  volume = {3},
  pages = {132},
  month = apr,
  year = {2019}
}

@article{tripartite2011looi,
  title = {Tripartite entanglement in qudit stabilizer states and application in quantum error correction},
  author = {Looi, Shiang Yong and Griffiths, Robert B.},
  journal = {Phys. Rev. A},
  volume = {84},
  issue = {5},
  pages = {052306},
  numpages = {14},
  year = {2011},
  month = {Nov},
  publisher = {American Physical Society},
  doi = {10.1103/PhysRevA.84.052306},
  url = {https://link.aps.org/doi/10.1103/PhysRevA.84.052306}
}

@misc{fattal2004entanglement,
      title={Entanglement in the stabilizer formalism}, 
      author={David Fattal and Toby S. Cubitt and Yoshihisa Yamamoto and Sergey Bravyi and Isaac L. Chuang},
      year={2004},
      eprint={quant-ph/0406168},
      archivePrefix={arXiv},
      primaryClass={quant-ph},
      url={https://arxiv.org/abs/quant-ph/0406168}, 
}

@article{brayvi2006ghz,
    author = {Bravyi, Sergey and Fattal, David and Gottesman, Daniel},
    title = {{GHZ} extraction yield for multipartite stabilizer states},
    journal = {Journal of Mathematical Physics},
    volume = {47},
    number = {6},
    pages = {062106},
    year = {2006},
    month = {06},
    issn = {0022-2488},
    doi = {10.1063/1.2203431},
    url = {https://doi.org/10.1063/1.2203431}
    
}

@article{dur2000three,
  title = {Three qubits can be entangled in two inequivalent ways},
  author = {D{\"u}r, Wolfgang and Vidal, Guifre and Cirac, J Ignacio},
  journal = {Physical Review A},
  volume = {62},
  number = {6},
  pages = {062314},
  year = {2000},
  publisher=  {APS},
  doi = {10.1103/PhysRevA.62.062314},
  url = {https://link.aps.org/doi/10.1103/PhysRevA.62.062314}
}

@article{ji2007lu,
  title={The LU-LC conjecture is false},
  author={Ji, Zhengfeng and Chen, Jianxin and Wei, Zhaohui and Ying, Mingsheng},
  journal={arXiv preprint arXiv:0709.1266},
  year={2007},
  url = {https://arxiv.org/abs/0709.1266},
}

@article{vidal1999robustness,
   title={Robustness of entanglement},
   author={Vidal, Guifr{\'e} and Tarrach, Rolf},
   journal={Phys. Rev. A},
   volume={59},
   issue={1},
   pages={141--155},
   year={1999},
   month={Jan},
   publisher={American Physical Society},
   doi={10.1103/PhysRevA.59.141},
   url={https://link.aps.org/doi/10.1103/PhysRevA.59.141}
}

@article{shannon1948mathematical,
   author={Shannon, C. E.},
   title={A mathematical theory of communication},
   journal={The Bell System Technical Journal},
   year={1948},
   volume={27},
   number={3},
   pages={379-423},
   doi={10.1002/j.1538-7305.1948.tb01338.x},
   url = {https://doi.org/10.1002/j.1538-7305.1948.tb01338.x}
}

@article{leone2022stabilizer,
  title = {Stabilizer R\'enyi Entropy},
  author = {Leone, Lorenzo and Oliviero, Salvatore F. E. and Hamma, Alioscia},
  journal = {Phys. Rev. Lett.},
  volume = {128},
  issue = {5},
  pages = {050402},
  numpages = {5},
  year = {2022},
  month = {Feb},
  publisher = {American Physical Society},
  doi = {10.1103/PhysRevLett.128.050402},
  url = {https://link.aps.org/doi/10.1103/PhysRevLett.128.050402}
}

@article{aaronson2004improved,
  title = {Improved simulation of stabilizer circuits},
  author = {Aaronson, Scott and Gottesman, Daniel},
  journal = {Phys. Rev. A},
  volume = {70},
  issue = {5},
  pages = {052328},
  numpages = {14},
  year = {2004},
  month = {Nov},
  publisher = {American Physical Society},
  doi = {10.1103/PhysRevA.70.052328},
  url = {https://link.aps.org/doi/10.1103/PhysRevA.70.052328}
}

@phdthesis{gottesman1997stabilizer,
  author       = {Gottesman, Daniel},
  title        = {Stabilizer Codes and Quantum Error Correction},
  school       = {California Institute of Technology},
  year         = {1997},
  type         = {Dissertation (Ph.D.)},
  doi          = {10.7907/rzr7-dt72},
  url          = {https://resolver.caltech.edu/CaltechETD:etd-07162004-113028},
  grantor      = {California Institute of Technology},
}

@article{bravyi2005universal,
  title = {Universal quantum computation with ideal Clifford gates and noisy ancillas},
  author = {Bravyi, Sergey and Kitaev, Alexei},
  journal = {Phys. Rev. A},
  volume = {71},
  issue = {2},
  pages = {022316},
  numpages = {14},
  year = {2005},
  month = {Feb},
  publisher = {American Physical Society},
  doi = {10.1103/PhysRevA.71.022316},
  url = {https://link.aps.org/doi/10.1103/PhysRevA.71.022316}
}

@article{gross2007non,
   author={Gross, D.},
   title={Non-negative Wigner functions in prime dimensions},
   journal={Applied Physics B},
   year={2007},
   volume={86},
   number={3},
   pages={367-370},
   sn={1432-0649},
   doi={10.1007/s00340-006-2510-9},
   url={https://doi.org/10.1007/s00340-006-2510-9}
}

@article{gross2006hudson,
   author={Gross, D.},
   title={Hudson’s theorem for finite-dimensional quantum systems},
   journal={Journal of Mathematical Physics},
   volume={47},
   number={12},
   pages={122107},
   year={2006},
   month={12},
   issn={0022-2488},
   doi={10.1063/1.2393152},
   url={https://doi.org/10.1063/1.2393152},
   
}

@article{wang2020efficiently,
  title = {Efficiently Computable Bounds for Magic State Distillation},
  author = {Wang, Xin and Wilde, Mark M. and Su, Yuan},
  journal = {Phys. Rev. Lett.},
  volume = {124},
  issue = {9},
  pages = {090505},
  numpages = {7},
  year = {2020},
  month = {Mar},
  publisher = {American Physical Society},
  doi = {10.1103/PhysRevLett.124.090505},
  url = {https://link.aps.org/doi/10.1103/PhysRevLett.124.090505}
}

@book{nielsen2010quantum,
  title={Quantum computation and quantum information},
  author={Nielsen, Michael A and Chuang, Isaac L},
  year={2010},
  publisher={Cambridge university press}
}

@article{hoggar199864,
   author={Hoggar, Stuart G.},
   title={64 Lines from a Quaternionic Polytope},
   journal={Geometriae Dedicata},
   year={1998},
   volume={69},
   number={3},
   pages={287-289},
   sn={1572-9168},
   doi={10.1023/A:1005009727232},
   url={https://doi.org/10.1023/A:1005009727232}
}

@article{howard2017application,
  title = {Application of a Resource Theory for Magic States to Fault-Tolerant Quantum Computing},
  author = {Howard, Mark and Campbell, Earl},
  journal = {Phys. Rev. Lett.},
  volume = {118},
  issue = {9},
  pages = {090501},
  numpages = {6},
  year = {2017},
  month = {Mar},
  publisher = {American Physical Society},
  doi = {10.1103/PhysRevLett.118.090501},
  url = {https://link.aps.org/doi/10.1103/PhysRevLett.118.090501}
}

@article{Hamaguchi2024Robustness,
  doi = {10.22331/q-2024-09-05-1461},
  url = {https://doi.org/10.22331/q-2024-09-05-1461},
  title = {Handbook for {Q}uantifying {R}obustness of {M}agic},
  author = {Hamaguchi, Hiroki and Hamada, Kou and Yoshioka, Nobuyuki},
  journal = {{Quantum}},
  issn = {2521-327X},
  publisher = {{Verein zur F{\"{o}}rderung des Open Access Publizierens in den Quantenwissenschaften}},
  volume = {8},
  pages = {1461},
  month = sep,
  year = {2024}
}

@article{PhysRevLett.128.210502,
  title = {Quantifying Qubit Magic Resource with {G}ottesman-{K}itaev-{P}reskill Encoding},
  author = {Hahn, Oliver and Ferraro, Alessandro and Hultquist, Lina and Ferrini, Giulia and Garc\'{\i}a-\'Alvarez, Laura},
  journal = {Phys. Rev. Lett.},
  volume = {128},
  issue = {21},
  pages = {210502},
  numpages = {7},
  year = {2022},
  month = {May},
  publisher = {American Physical Society},
  doi = {10.1103/PhysRevLett.128.210502},
  url = {https://link.aps.org/doi/10.1103/PhysRevLett.128.210502}
}

@article{mele2024introduction,
  doi = {10.22331/q-2024-05-08-1340},
  url = {https://doi.org/10.22331/q-2024-05-08-1340},
  title = {Introduction to {H}aar {M}easure {T}ools in {Q}uantum {I}nformation: {A} {B}eginner's {T}utorial},
  author = {Mele, Antonio Anna},
  journal = {{Quantum}},
  issn = {2521-327X},
  publisher = {{Verein zur F{\"{o}}rderung des Open Access Publizierens in den Quantenwissenschaften}},
  volume = {8},
  pages = {1340},
  month = may,
  year = {2024}
}

@article{chen2024magic,
  title = {Magic of random matrix product states},
  author = {Chen, Liyuan and Garcia, Roy J. and Bu, Kaifeng and Jaffe, Arthur},
  journal = {Phys. Rev. B},
  volume = {109},
  issue = {17},
  pages = {174207},
  numpages = {12},
  year = {2024},
  month = {May},
  publisher = {American Physical Society},
  doi = {10.1103/PhysRevB.109.174207},
  url = {https://link.aps.org/doi/10.1103/PhysRevB.109.174207}
}

@article{turkeshi2023pauli,
  title = {Pauli spectrum and nonstabilizerness of typical quantum many-body states},
  author = {Turkeshi, Xhek and Dymarsky, Anatoly and Sierant, Piotr},
  journal = {Phys. Rev. B},
  volume = {111},
  issue = {5},
  pages = {054301},
  numpages = {12},
  year = {2025},
  month = {Feb},
  publisher = {American Physical Society},
  doi = {10.1103/PhysRevB.111.054301},
  url = {https://link.aps.org/doi/10.1103/PhysRevB.111.054301}
}

@article{turkeshi2025magic,
   author={Turkeshi, Xhek and Tirrito, Emanuele and Sierant, Piotr},
   title={Magic spreading in random quantum circuits},
   journal={Nature Communications},
   year={2025},
   volume={16},
   number={1},
   pages={2575},
   doi={10.1038/s41467-025-57704-x},
   url={https://doi.org/10.1038/s41467-025-57704-x},
   sn={2041-1723},

}

@misc{qiskit2024,
      title={Quantum computing with {Q}iskit},
      author={Javadi-Abhari, Ali and Treinish, Matthew and Krsulich, Kevin and Wood, Christopher J. and Lishman, Jake and Gacon, Julien and Martiel, Simon and Nation, Paul D. and Bishop, Lev S. and Cross, Andrew W. and Johnson, Blake R. and Gambetta, Jay M.},
      year={2024},
      doi={10.48550/arXiv.2405.08810},
      eprint={2405.08810},
      archivePrefix={arXiv},
      primaryClass={quant-ph}
}

@misc{iannotti2025entanglement,
   title={Entanglement and Stabilizer entropies of random bipartite pure quantum states},
   author={Daniele Iannotti and Gianluca Esposito and Lorenzo Campos Venuti and Alioscia Hamma},
   year={2025},
   eprint={2501.19261},
   archivePrefix={arXiv},
   primaryClass={quant-ph},
   url={https://arxiv.org/abs/2501.19261}
}

@article{walter2016multipartite,
  title={Multipartite entanglement},
  author={Walter, Michael and Gross, David and Eisert, Jens},
  journal={Quantum Information: From Foundations to Quantum Technology Applications},
  pages={293--330},
  year={2016},
  publisher={Wiley Online Library},
  doi = {https://doi.org/10.1002/9783527805785.ch14}
}

@article{PhysRevA.80.042309,
  title = {Quantum circuit for three-qubit random states},
  author = {Giraud, Olivier and \ifmmode \check{Z}\else \v{Z}\fi{}nidari\ifmmode \check{c}\else \v{c}\fi{}, Marko and Georgeot, Bertrand},
  journal = {Phys. Rev. A},
  volume = {80},
  issue = {4},
  pages = {042309},
  numpages = {7},
  year = {2009},
  month = {Oct},
  publisher = {American Physical Society},
  doi = {10.1103/PhysRevA.80.042309},
  url = {https://link.aps.org/doi/10.1103/PhysRevA.80.042309}
}

@article{audenaert2007sharp,
  title={A sharp continuity estimate for the von Neumann entropy},
  author={Audenaert, Koenraad MR},
  journal={Journal of Physics A: Mathematical and Theoretical},
  volume={40},
  number={28},
  pages={8127},
  year={2007},
  publisher={IOP Publishing},
  doi = {10.1088/1751-8113/40/28/S18},
  url = {https://dx.doi.org/10.1088/1751-8113/40/28/S18},
}

@article{lambert2024qutip,
  title={QuTiP 5: The Quantum Toolbox in Python},
  author={Lambert, Neill and Gigu{\`e}re, Eric and Menczel, Paul and Li, Boxi and Hopf, Patrick and Su{\'a}rez, Gerardo and Gali, Marc and Lishman, Jake and Gadhvi, Rushiraj and Agarwal, Rochisha and others},
  journal={arXiv preprint arXiv:2412.04705},
  url={https://arxiv.org/abs/2412.04705},
  year={2024}
}

@article{seddon2021quantifying,
  title={Quantifying quantum speedups: Improved classical simulation from tighter magic monotones},
  author={Seddon, James R and Regula, Bartosz and Pashayan, Hakop and Ouyang, Yingkai and Campbell, Earl T},
  journal={PRX Quantum},
  volume={2},
  number={1},
  pages={010345},
  year={2021},
  publisher={APS},
  url={https://journals.aps.org/prxquantum/abstract/10.1103/PRXQuantum.2.010345}
}

\end{document}